\documentclass[11pt]{article}


\usepackage{amsmath,amssymb,amsthm,mathtools}

    \usepackage{thmtools}
    \usepackage{thm-restate}
	\usepackage{a4,geometry}

\usepackage{graphicx}
\usepackage{paralist}
\usepackage{bm}
\usepackage{xspace}
\usepackage{url}
\usepackage{fullpage, prettyref}
\usepackage{boxedminipage}
\usepackage{wrapfig}
\usepackage{ifthen}
\usepackage{color}
\usepackage[usenames,dvipsnames]{xcolor}
\usepackage[colorlinks,citecolor=blue,linkcolor=BrickRed]{hyperref}\usepackage[colorlinks,citecolor=blue,linkcolor=BrickRed]{hyperref}
\usepackage{framed}

\usepackage{algpseudocode}
\usepackage[ruled,vlined, linesnumbered]{algorithm2e}
\usepackage{titlesec}
\titlespacing*{\subsubsection}{0pt}{0.2em}{0.2em}

\usepackage{thmtools}
\usepackage{thm-restate}
\usepackage{cleveref}

\newtheorem{theorem}{Theorem}[section]

\newtheorem{lemma}[theorem]{Lemma}

\newtheorem{question}[theorem]{Question}
\newtheorem{corollary}[theorem]{Corollary}
\newtheorem{definition}[theorem]{Definition}
\newtheorem{proposition}[theorem]{Proposition}
\newtheorem{fact}[theorem]{Fact}

\newtheorem{remark}[theorem]{Remark}
\newtheorem{abort}[theorem]{Stopping Condition}
\newtheorem{mergerule}[theorem]{Merging Rule}

\newcommand{\ignore}[1]{}

\newcommand{\p}{\Pr}
\newcommand{\R}{\mathbb R}

\newcommand{\poly}{\mathsf{poly}}
\newcommand{\polylog}{\mathsf{polylog}}

\newcommand{\diag}{\mathsf{diag}}
\newcommand{\Cov}{\mathsf{Cov}}

\newcommand{\GlobalIntervalTree}{\mathcal{T}_{\mathsf{Global}}}

\newcommand{\ModifiedGlobalIntervalTree}{\mathcal{T}_{\mathsf{Global}}^{\infty}}

\newcommand{\E}{\mathbb{E}}

\newcommand\restr[2]{{
  \left.\kern-\nulldelimiterspace
  #1 
  \vphantom{\big|} 
  \right|_{#2} 
 }}

\newcommand{\Sec}[1]{\hyperref[sec:#1]{\S\ref*{sec:#1}}} 
\newcommand{\Eqn}[1]{\hyperref[eq:#1]{(\ref*{eq:#1})}} 
\newcommand{\Fig}[1]{\hyperref[fig:#1]{Fig.\,\ref*{fig:#1}}}
\newcommand{\Tab}[1]{\hyperref[tab:#1]{Tab.\,\ref*{tab:#1}}}
\newcommand{\Thm}[1]{\hyperref[thm:#1]{Theorem\,\ref*{thm:#1}}} 
\newcommand{\Fact}[1]{\hyperref[fact:#1]{Fact\,\ref*{fact:#1}}} 
\newcommand{\Lem}[1]{\hyperref[lem:#1]{Lemma~\ref*{lem:#1}}} 
\newcommand{\Prop}[1]{\hyperref[prop:#1]{Proposition~\ref*{prop:#1}}} 
\newcommand{\Cor}[1]{\hyperref[cor:#1]{Corollary~\ref*{cor:#1}}} 
\newcommand{\Conj}[1]{\hyperref[conj:#1]{Conjecture~\ref*{conj:#1}}} 
\newcommand{\Rem}[1]{\hyperref[rem:#1]{Remark~\ref*{rem:#1}}} 
\newcommand{\Def}[1]{\hyperref[def:#1]{Definition~\ref*{def:#1}}} 
\newcommand{\Alg}[1]{\hyperref[alg:#1]{Alg.~\ref*{alg:#1}}} 
\newcommand{\Ex}[1]{\hyperref[ex:#1]{Ex.~\ref*{ex:#1}}} 
\newcommand{\Clm}[1]{\hyperref[clm:#1]{Claim~\ref*{clm:#1}}} 
\newcommand{\Step}[1]{\hyperref[step:#1]{Step~\ref*{step:#1}}} 
\newcommand{\Obs}[1]{\hyperref[obs:#1]{Observation~\ref*{obs:#1}}} 
\renewcommand{\Alg}[1]{\hyperref[alg:#1]{Algorithm~\ref*{alg:#1}}} 

\newcommand{\disc}{\mathsf{disc}}

\newcommand{\err}{\mathsf{err}}
\newcommand{\ASI}{\mathsf{ASI}}
\newcommand{\Error}{\mathsf{Error}}
\newcommand{\Bad}{\mathsf{Bad}}
\newcommand{\bad}{\mathsf{bad}}

\graphicspath{{./Figures/}}

\newcommand{\Tr}{\text{Tr}}

\title{Near-Optimal Constructive Bounds for $\ell_2$ Prefix Discrepancy and Steinitz Problems via Affine Spectral Independence}

\date{}

\author{
Kunal Dutta\thanks{University of Warsaw, Warsaw, Poland. \texttt{kdutta@mimuw.edu.pl}. Supported by the Polish NCN OPUS grant nr. 2023/51/B/ST6/02989.}
\and 
Agastya Vibhuti Jha\thanks{University of Chicago, Chicago, IL, USA. \texttt{agastyajha@uchicago.edu}.}
\and
Haotian Jiang\thanks{University of Chicago, Chicago, IL, USA. \texttt{jhtdavid@uchicago.edu}.}}

\begin{document}

\allowdisplaybreaks
\begin{titlepage}
\maketitle

\begin{abstract}
A classical result of Steinitz from 1913 \cite{Ste13}, answering an earlier question of Riemann and L\'evy (e.g., \cite{Lev05}), states that for any norm $\|\cdot\|$ in $\mathbb{R}^d$ and any set of vectors $v_1, \cdots, v_n \in \R^d$ satisfying $\sum_{i=1}^n v_i = 0$, there exists an ordering $\pi: [n] \rightarrow [n]$ such that every partial sum along this order is bounded by $O(d)$, i.e., $\big\| \sum_{i=1}^t v_{\pi(i)} \big\| \leq O(d)$ for all $t \in [n]$. 

Steinitz's bound is tight up to constants in general, but for the $\ell_2$ norm $\|\cdot\|_2$, it has been conjectured that the best bound is $O(\sqrt{d})$.
Almost a century later, a breakthrough work of Banaszczyk \cite{Ban12} gave a bound of $O(\sqrt{d} + \sqrt{\log  n})$ for the $\ell_2$ Steinitz problem, matching the conjecture under the mild assumption that $d \geq \Omega(\log n)$. Banaszczyk's result is non-constructive, and the previous best algorithmic bound was $O(\sqrt{d \log n})$, due to Bansal and Garg \cite{BG17}. 

In this work, we give an efficient algorithm that matches the conjectured $O(\sqrt{d})$ bound for the $\ell_2$ Steinitz problem under the slightly worse, yet still polylogarithmic, condition of $d \geq \Omega(\log^7 n)$. As in prior work, our result extends to the harder problem of $\ell_2$ prefix discrepancy.  

We employ the framework of obtaining the desired ordering via a discrete Brownian motion, guided by a semidefinite program (SDP). To obtain our results, we use the new technique of ``Decoupling via Affine Spectral Independence'', proposed by Bansal and Jiang \cite{BJ26} to achieve substantial progress on the Beck-Fiala and Koml\'os conjectures, together with a ``Global Interval Tree'' data structure that simultaneously controls the deviations for all prefixes. 
\end{abstract}

 \thispagestyle{empty}
\end{titlepage}

\thispagestyle{empty}
{\hypersetup{linkcolor=BrickRed}
 \tableofcontents
}

\thispagestyle{empty}
\newpage
\setcounter{page}{1}

\section{Introduction}
The Steinitz problem is a fundamental question in combinatorial discrepancy theory that studies how well one can control the discrepancy of partial sums of vectors under permutations. 
Formally, given a norm $\|\cdot\|$ and a set of vectors $v_1, \cdots, v_n \in \R^d$ satisfying $\|v_i\| \leq 1$ for each $i \in [n]$\footnote{Throughout, we use $[n]$ to denote the set $\{1, \cdots, n\}$.} and $\sum_{i \in [n]} v_i = 0$, the Steinitz problem asks for the smallest number $S(\| \cdot \|)$, depending only on the norm $\|\cdot\|$, such that there exists an ordering $\pi: [n] \rightarrow [n]$ for which all partial sums along this ordering are at most $S(\|\cdot\|)$, i.e., $\max_{t \in [n]} \big\|\sum_{i \in [t]} v_{\pi(i)}\big\| \leq S(\|\cdot\|)$.  
This was originally a question of Riemann and L\'evy \cite{Lev05} from the early 1900s, and it has later found many surprising applications in areas such as graph theory \cite{AB86}, integer programming \cite{BMMP12,DFG12,EW19,JR19}, and scheduling \cite{Bar81,Sev94}. 
We refer readers to the surveys \cite{HA89,Sev94,Bar08} for more details on the history and applications of the Steinitz problem.

A classical result of Steinitz \cite{Ste13} showed that $S(\|\cdot\|) \leq 2d$, regardless of the norm. Subsequently, this bound was improved to $1.5d$ in \cite{Bar81} and then further to slightly better than $d$ \cite{GS80,Ban87}. 
In particular, the proof by Grinberg and Sevastyanov in \cite{GS80} uses a clever iterated rounding argument, which also implies an efficient algorithm for finding a good ordering $\pi$.  

The above bound of $O(d)$ is the best possible up to constants for general norms, e.g., $S(\|\cdot\|_1) \geq d/2$. But for the $\ell_2$ norm, the best lower bound is only $S(\|\cdot\|_2) \geq \Omega(\sqrt{d})$ \cite{Ber31,GS80}. While no  $o(d)$ bound is known for $\ell_2$, it was conjectured that $\Theta(\sqrt{d})$ should be the right answer.\footnote{Note that one may assume $d \leq n^2$, as otherwise the trivial discrepancy bound of $n$ already achieves the conjectured $O(\sqrt{d})$ bound. We make the assumption $d \leq n^2$ throughout this work.} 
In breakthrough work, Banaszczyk \cite{Ban12} proved that $S(\|\cdot\|_2) \leq O(\sqrt{d} + \sqrt{\log n})$, matching the conjectured bound up to constants under the mild assumption of $d \geq \log n$. 
Unlike the $O(d)$ bound above, Banaszczyk's proof is non-constructive and does not give an efficient algorithm for finding a good ordering (we discuss more on algorithmic aspects of discrepancy theory in \Cref{subsec:related_work}). 
On the algorithmic front, \cite{HS14} gave a constructive proof of $O(\sqrt{d} \log^{2.5} n)$, and the previous best algorithmic bound was $O(\sqrt{d \log n})$ due to Bansal and Garg \cite{BG17}. 
Unlike Banaszczyk's result, these algorithmic bounds don't match the conjectured bound for $\ell_2$ Steinitz in any regime, and it is a natural open question  to algorithmically attain Banaszczyk's bound for this problem (e.g., see \cite{BDGL18}).

\smallskip
\noindent \textbf{Prefix Discrepancy.}
A closely-related question is the $\ell_2$ {\em prefix discrepancy}\footnote{The prefix discrepancy problem is also referred to as the {\em signed series} problem in the literature.} problem, introduced by Spencer \cite{Spe77}: given a sequence of vectors $v_1, \cdots, v_n \in \mathbb{R}^d$ with $\|v_i\|_2 \leq 1$ for each $i \in [n]$, the goal is to find a coloring $x \in \{\pm 1\}^n$ so that the $\ell_2$ norm of all partial signed sum is bounded by a number $E(\|\cdot\|_2)$, i.e., $\max_{t \in [n]} \big\| \sum_{i=1}^t x_i v_t \big\|_2 \leq E(\|\cdot\|_2)$.
Variants of this prefix problem have been well-studied and they have many applications to online algorithms \cite{JKS19,BJSS20,BJM+21,ALS21}, Tusn\'ady's problem \cite{Nik17}, and flow time scheduling \cite{BRS22}.

For $\ell_2$ prefix discrepancy, Spencer \cite{Spe77} showed a bound that only depends on $d$. Using iterated rounding, B\'ar\'any and Grinberg \cite{BG81} gave an algorithmic bound of $2d-1$.
Similar to $\ell_2$ Steinitz, it has been conjectured that the correct bound for $E(\|\cdot\|_2)$ should be $\Theta(\sqrt{d})$.
The $\ell_2$ prefix discrepancy problem is actually known to be harder than the Steinitz problem --- Chobanyan \cite{Cho94} gave a reduction between the two problems which shows that $S(\|\cdot\|_2) \leq E(\|\cdot \|_2)$, and this reduction can be made algorithmic \cite{HS14}. In fact, both the non-constructive bound of $O(\sqrt{d} + \sqrt{\log n})$ in \cite{Ban12} and the algorithmic bound of $O(\sqrt{d \log n})$ in \cite{BG17} above were given for the harder $\ell_2$ prefix discrepancy problem and then they used this reduction as a black box.

\subsection{Our Results}

Our main result is an improved algorithmic bound for $\ell_2$ prefix discrepancy that matches the conjectured $\Theta(\sqrt{d})$ bound under the slightly worse, yet still polylogarithmic, condition of $d \geq \log^7 n$ (recall that in Banaszczyk's non-constructive bound, the condition was $d \geq \log n$). 

\begin{restatable}{theorem}{elltwotoelltwo}(Algorithmic $\ell_2$ to $\ell_2$ prefix discrepancy)
\label{thm:main-ell2-to-ell2}
Given $v_1, \cdots, v_n \in \R^d$ with $\|v_i\|_2 \leq 1$ for each $i \in [n]$, one can efficiently find $x \in \{\pm 1\}^n$ such that $\big\|\sum_{i=1}^t x_i v_i\big\|_2 \leq O(\sqrt{d} + d^{1/4} \log^{7/4} n)$ for all prefix $t \in [n]$. In particular, the bound matches the conjectured $\Theta(\sqrt{d})$ when $d \geq \log^7 n$. 
\end{restatable}

Unlike the constructive bound in \cite{BG17} that loses a multiplicative $O(\sqrt{\log n})$ factor, the bound in \Cref{thm:main-ell2-to-ell2} (as well as the other results below) is only off from the conjectured bound by an {\em additive} factor, which is of lower order when $d \geq \polylog(n)$. 
While not stated explicitly, our bounds are also never worse than those in \cite{BG17}, as our algorithms are built on top of theirs.

Using the algorithmic reduction in \cite{Cho94,HS14} as a black box, \Cref{thm:main-ell2-to-ell2} also implies the same algorithmic improvement for the $\ell_2$ Steinitz problem. This gets close to answering the question of algorithmically matching Banaszczyk's bound for $\ell_2$ Steinitz in \cite{BDGL18}. 

\begin{corollary}[Algorithmic $\ell_2$ Steinitz] 
\label{thm:main-ell2-Steinitz}
Given vectors $v_1, \cdots, v_n \in \R^d$ such that $\sum_{i=1}^n v_i = 0$ and $\|v_i\|_2 \leq 1$ for each $i \in [n]$, one can efficiently find a permutation $\pi: [n] \rightarrow [n]$ such that $\big\|\sum_{i=1}^t v_{\pi(i)}\big\|_2 \leq O(\sqrt{d} + d^{1/4} \log^{7/4} n)$ for all prefix $t \in [n]$. This bound matches the conjectured $\Theta(\sqrt{d})$ bound when $d \geq \log^7 n$.
\end{corollary}

It is an intriguing open question to improve the condition of $d \geq \log^7 n$ in Theorems \ref{thm:main-ell2-to-ell2} and \ref{thm:main-ell2-Steinitz} to get close to the condition of $d \geq \log n$ in Banaszczyk's non-constructive bound \cite{Ban12}. 
We give such an improvement in the case where coordinates of the vectors $v_i$ are ``well-spread'', namely, each coordinate of every $v_i$ has magnitude at most $1/\sqrt{d}$. 
Equivalently, one may scale the vectors $v_1, \cdots, v_n$ up by a $\sqrt{d}$ factor and assume that each $v_i$ satisfies $\|v_i\|_\infty \leq 1$. 
Under this scaling, the conjectured bound becomes $\Theta(d)$, where the lower bound is witnessed by a Hadamard matrix. 
The previous best constructive bound for this setting was $O(d \sqrt{\log n})$ by Bansal and Garg \cite{BG17}.

\begin{restatable}{theorem}{ellinftoelltwo}(Algorithmic $\ell_\infty$ to $\ell_2$ prefix discrepancy)
\label{thm:ell_inf-to-ell_2}
Given $v_1, \cdots, v_n \in \R^d$ with $\|v_i\|_\infty \leq 1$ for each $i \in [n]$, one can efficiently find $x \in \{\pm 1\}^n$ such that $\big\|\sum_{i=1}^t x_i v_i\big\|_2 \leq O(d + d^{3/4} \log n + d^{1/4}\log^{3/2}n)$ for all prefix $t \in [n]$. This matches the conjectured $\Theta(d)$ bound when $d \geq \log^4 n$. 
\end{restatable}

\subsection{Our Approach in a Nutshell}  
\label{subsec:nutshell}

We view the given vectors $v_1, \cdots, v_n \in \R^d$ as a matrix $A \in \R^{d \times n}$ whose $j$th column is $v_j$. Denote by $A_i \in \R^n$ its $i$th row. 
Throughout, we use $v(\ell)$ to denote the $\ell$th entry of a vector $v$ and $\log$ means logarithm to the base $2$. 
We always use $i \in [d]$ to index rows of $A$, and $j \in [n]$ for columns, and use $\mathcal{P} \subseteq [n]$ to denote an arbitrary prefix whose discrepancy we will bound.

Similar to prior works, our algorithm starts with $x_0 = 0^n$ and evolves a fractional coloring $x_t \in [-1,1]^n$ using carefully chosen random tiny increments $\Delta x_t$. 
As in \cite{BG17}, to handle all prefixes, our algorithm maintains a {\em sliding window} $W_t \subseteq [n]$ containing the first $10 d$ {\em alive} columns $j$ for which $x_t(j)$ has not yet reached $\pm 1$, and ensures that (1) $\Delta x_t$ is supported only on $W_t$, and (2) the discrepancy update of the full sliding window is $0$, i.e., $\sum_{j \in W_t} A_i(j) \Delta x_t(j) = 0$ for every row $i \in [d]$. 
This guarantees that each prefix $\mathcal{P} \in [n]$ only incurs non-zero discrepancy while $\mathcal{P} \in W_t$.

Roughly speaking, prior algorithms \cite{BDG19,BG17,BLV22} choose the increment $\Delta x_t$ to be $O(1)$-spectrally independent, which ensures that the discrepancy $\varphi_t^{\mathcal{P}}(i)$ for each prefix $\mathcal{P} \in [n]$ and each row $i \in [d]$ is $O(d)$-subgaussian.\footnote{A mean-zero random vector $\varphi \in \mathbb{R}^m$ is called $\sigma^2$-subgaussian, if for any unit test vector $\theta \in \R^m$, 
it satisfies the Gaussian-type tail bound $\p(|\langle \theta, \varphi\rangle| \geq t) \leq 2\exp(-t^2/(2\sigma^2))$.}
This suffices for Bansal-Garg's $O(d \sqrt{\log n})$ bound for $\ell_\infty$ to $\ell_2$ prefix discrepancy --- subgaussianity ensures that each row prefix discrepancy $|\varphi_t^{\mathcal{P}}(i)| \leq O(\sqrt{d \log n})$ with probability $1- 1/\poly(n)$, which implies $\|\varphi_t^{\mathcal{P}}\|_2 \leq O(d\sqrt{ \log n})$ by taking a union bound over all rows $i \in [d]$ and prefixes $\mathcal{P} \in [n]$. A similar argument can also recover\footnote{The original proof of the $\ell_2$ prefix discrepancy bound in \cite{BG17} was done by analyzing the dynamics of $d \|\varphi_t^{\mathcal{P}}\|_2^2$, but a row-wise analysis already recovers their bound and is  technically slightly simpler.}  Bansal-Garg's bound for $\ell_2$ to $\ell_2$ prefix discrepancy, but we will focus on the $\ell_\infty$ to $\ell_2$ setting here for simplicity.

The main bottleneck for Bansal-Garg's bound is that while each $\varphi_t^{\mathcal{P}}(i)$ is $O(d)$-subgaussian, they may have arbitrarily correlations, and this is why a row-wise analysis is needed to control $\varphi_t^{\mathcal{P}}(i)$ up to a $1/\poly(n)$ tail probability (and hence loses a $O(\sqrt{\log n})$ multiplicative factor) to union bound over all rows and prefixes. 
To avoid this loss, we observe that if the $\varphi_t^{\mathcal{P}}(i)$ were mutually independent, standard concentration inequalities would 
give an {\em additive} deviation, and recover Banaszczyk's non-constructive $O(d + \sqrt{d\log n})$ bound.
Although it is clearly impossible to achieve full independence for $\varphi_t^{\mathcal{P}}(i)$, we leverage the recent technique of ``Decoupling via Affine Spectral Independence'' in \cite{BJ26} to ensure that the change in prefix discrepancy $\Delta \varphi_t^{\mathcal{P}}(i)$ is {\em affine spectrally independent} --- this weaker form of independence turns out to be sufficient for an additive $\ell_2$ discrepancy bound on $\varphi_t^{\mathcal{P}}$, albeit with a worse additive term (than the fully independent case). 

Nonetheless, the use of affine spectral independence posts a technical challenge --- the independence degrades as one controls more rows/prefixes. Naively, at each time $t$, one needs to control $d$ rows for each of the $|W_t| = 10d$ prefixes in the sliding window, but including all these $\Theta(d^2)$ row-prefix pairs would result in too weak of an independence to even improve upon the Bansal-Garg bound. 
To bypass this issue, we only enforce affine spectral independence for a carefully chosen set of $o(d)$ prefixes in $W_t$ called {\em ASI-guarded prefixes}. 
To maintain this set of ASI-guarded prefixes dynamically and to bound the discrepancies of the remaining {\em unguarded} prefixes, we design a {\em Global Interval Tree} data structure to simultaneously control their additive deviations.

\smallskip
\noindent \textbf{Roadmap.}
The rest of this paper is organized as follows. 
Some further related work will be discussed in \Cref{subsec:related_work}.  We then give an overview of our approach in \Cref{sec:overview}. A formal presentation of our algorithmic framework, data structure, and meta analysis will appear in \Cref{sec:framework}. 
Then in \Cref{sec:ell_inf-to-ell_2}, we prove our result for $\ell_\infty$ to $\ell_2$ prefix discrepancy in \Cref{thm:ell_inf-to-ell_2}.
Finally, we prove our $\ell_2$ to $\ell_2$ prefix discrepancy bound in \Cref{thm:main-ell2-to-ell2} in \Cref{sec:ell2-to-ell2}.
We will conclude and mention some open problems in \Cref{sec:conclusion_open_problems}.
Some missing details will be given in the appendix. 

\subsection{Further Related Work}
\label{subsec:related_work}

\smallskip
\noindent \textbf{Combinatorial Discrepancy.} 
Combinatorial discrepancy theory is a well-studied topic with many connections and applications to both mathematics and computer science, and we refer readers to the textbooks \cite{Cha00,Mat09,CST14}.
A classical question in combinatorial discrepancy is the following vector balancing problem: given a norm $\|\cdot\|$ and vectors $v_1, \cdots, v_n \in \mathbb{R}^d$, the goal is to find a coloring $x \in \{\pm 1\}^n$ to minimize $\disc_{\|\cdot\|}(v_1, \cdots, v_n) := \min_{x \in \{\pm 1\}^n} \big\| \sum_{i=1}^n x_i v_i \big\|$. 
This question is easier than prefix discrepancy, as one only bounds the total signed sum.
For instance, for the $\ell_2$ norm, a classical result of B\'ar\'any and Grinberg \cite{BG81} shows that $\disc_{\|\cdot \|_2}(v_1, \cdots, v_n) \leq O(\sqrt{d})$ when each $\|v_i\|_2 \leq 1$ and this bound is tight up to constants.

\smallskip
\noindent \textbf{$\ell_\infty$ Discrepancy and Algorithmic Aspects.} While our focus here is on $\ell_2$ discrepancy, we remark that $\ell_\infty$ discrepancy is also very natural, and its study has led to many interesting techniques and  developments in discrepancy theory. 
A seminal result of Spencer \cite{Spe85} (and independently, Gluskin \cite{Glu89}) shows that if each $\|v_i\|_\infty \leq 1$ and $d = n$, then $\disc_{\|\cdot\|_\infty}(v_1, \cdots, v_n) \leq O(\sqrt{n})$. This beats the $O(\sqrt{n \log n})$ bound for a random coloring, and is optimal up to constants. 

When each $\|v_i\|_2 \leq 1$, the long-standing Koml\'os conjecture, generalizing a seminal conjecture of Beck and Fiala \cite{BF81}, asserts that $\disc_{\|\cdot\|_\infty}(v_1, \cdots, v_n) \leq O(1)$. 
For the Koml\'os problem, Spencer \cite{Spe85} gave an $O(\log n)$ bound. This was improved to $O(\sqrt{\log n})$ by Banaszczyk \cite{Ban98}. His results for prefix discrepancy and Steinitz problems in \cite{Ban12} are crucially based on this earlier work.

While the above approaches for $\ell_\infty$ discrepancy have become prominent in discrepancy theory, these methods were originally non-constructive and did not give an efficient algorithm for finding a good coloring. Following a breakthrough of Bansal \cite{Ban10}, many elegant algorithms have been developed that matches these non-constructive bounds \cite{BS13,HSS14,LM15,DGLN16,Rot17,  LRR17,ES18,BDGL18,BDG19,ALS21,BLV22,JSS23, PV23,HSSZ24,BJ25a}, leading to surprising applications in many different areas (e.g., see the excellent survey \cite{Ban22}). Interestingly, they also provide new insights into discrepancy theory and lately, building on top of these algorithmic developments, Bansal and Jiang \cite{BJ25b,BJ26} gave an improved $O((\log^{1/4} n) \cdot \poly(\log\log n))$ bound for the Koml\'os problem and almost resolved the Beck-Fiala conjecture. 

\smallskip
\noindent \textbf{$\ell_\infty$ Prefix Discrepancy.} The $\ell_\infty$ prefix discrepancy problem, where the objective is to bound $\max_{t \in [n]} \|\sum_{i=1}^t x_i v_i\|_\infty$, is also well-studied. 
When each $\|v_i\|_\infty \leq 1$, it was conjectured that the correct bound should be $\Theta(\sqrt{d})$. 
Banaszczyk \cite{Ban12} gave a non-constructive bound of $O(\sqrt{d \log n})$, and this bound was later matched algorithmically in \cite{BG17}. 

The case where each $\|v_i\|_2 \leq 1$ (i.e., the prefix version of Koml\'os problem) is more elusive. The best non-constructive bound is $O(\sqrt{\log n})$ \cite{Ban12}, while the best algorithmic bound is only $O(\log n)$ \cite{ALS21}. It was asked in \cite{BJM+21} whether the Koml\'os conjecture generalizes to this prefix setting, i.e., whether the best bound is $O(1)$, but there is no clear consensus on this question.

\section{Overview of Our Approach}
\label{sec:overview}

In this section, we give an overview of our approach. Our focus here will be on the $\ell_\infty$ to $\ell_2$ prefix discrepancy bound in \Cref{thm:ell_inf-to-ell_2}, since it is technically simpler (than $\ell_2$ to $\ell_2$ prefix discrepancy) but already contains most of the ideas. 
We start with the algorithmic framework that we will use throughout the paper.  
Here we adopt the same notation as in \Cref{subsec:nutshell}.

\subsection{Algorithmic Framework: SDP-Guided Walk with A Sliding Window}
\label{subsec:basic_framework}
As in many previous discrepancy  algorithms \cite{BG17,BDG19,BLV22,Ban24,BJ25b,BJ26}, our algorithms start with $x_0 = 0^n$, and evolve a {\em fractional coloring} $x_t \in [-1,1]^n$ over time using tiny random increments, until some final coloring in $\{\pm 1\}^n$ is reached. The time $t$ will range from $0$ to $n$, and is updated in discrete increments of size $dt$. We will set $dt = 1/\poly(n)$ so that the algorithm runs in polynomial time, but it is useful to view $dt$ as infinitesimally small. 

\smallskip
\noindent \textbf{Sliding Window and Coloring Update.} At each time $t$, a column $j \in [n]$ is {\em alive} if $|x_t(j)| \leq 1- 1/(2nd)$, and is called {\em dead} otherwise. Notice that rounding a dead column to either $1$ or $-1$ incurs a negligible $O(1/(nd))$ discrepancy per column, and thus will be ignored henceforth. 
Let $W_t \subseteq [n]$ be the first $10d$ alive columns, which we will refer to as the {\em sliding window}.
The columns in $W_t$ will be called {\em active}, and alive columns not contained in $W_t$ are called {\em dormant}. The algorithm will only update the coloring $x_t(j)$ for $j \in W_t$, and ensures that the discrepancy update of the whole sliding window is $0$. 

Specifically, at time $t$, the algorithm chooses a random vector $u_t \in \R^{W_t}$ satisfying $\E[u_t]=0$, $\|u_t\|_2 = 1$, $x_t \perp u_t$, and $\sum_{j \in W_t} A_i(j) u_t(j) = 0$ for all rows $i \in [d]$ (and various other properties that will be specified later).\footnote{We sometimes abuse notation and also view $u_t$ as a vector in $\R^n$ with $u_t(j) = 0$ for all $j \notin W_t$.}
The coloring $x_t$ is updated to $x_{t+dt}= x_t + dx_t$ with
\[d x_t = u_t \sqrt{dt}.\]

Note that $\|u_t\|_2=1$ and $x_t \perp u_t$ ensures that $\|x_t\|_2^2=t$ for all $t$, so the process ends by $t=n$. 
Also notice that $\sum_{j \in W_t} A_i(j) u_t(j) = 0$ for all $i \in [d]$ ensures that the sliding window incurs $0$ discrepancy change, and thus any prefix $\mathcal{P} \in [n]$ only incurs non-zero discrepancy when $\mathcal{P} \in W_t$. 

\medskip

\begin{figure}[h]
    \centering
    \includegraphics[width=0.7\textwidth]{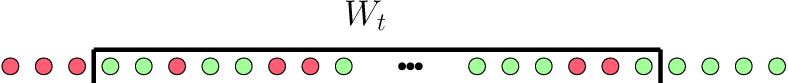}
    \caption{Sliding window. Green dots denote alive columns, and red dots denote dead ones.}
    \label{fig:TimeWindow}
\end{figure}

\medskip
\noindent{\bf Choosing $u_t$ via an SDP.}
The power of this framework comes from the flexibility to choose $u_t$ adaptively at each time $t$. 
To do this, the algorithm computes a PSD matrix $U_t$ by solving a semidefinite program (SDP), and samples $u_t$ with $U_t$ as the covariance matrix. 
As this approach of sampling $u_t$ via solving an SDP is already standard, to simplify the discussion here, we postpone the details of the SDP and how to sample $u_t$ from $U_t$ to \Cref{sec:framework}.

\subsection{Spectral Independence and the Bansal-Garg Bound}
\label{subsec:SI-BGbound-Overview}

A useful property ensured by the SDP in prior algorithms is that $u_t$ is $O(1)$-{\em spectrally independent}: 
\begin{align} \label{eq:SI-overview}
\Cov(u_t) \preceq O(1) \cdot \diag(\Cov(u_t)). \tag{$\mathsf{SI}$} 
\end{align}
Since $\Cov(u_t) = \E[u_t u_t^\top]$ for mean-zero $u_t$, this is equivalent to saying that for any vector $a \in \R^{W_t}$, 
\begin{align} \label{eq:pairwise-ind-overview}
 \E\left[\langle a, u_t \rangle ^2 \right] = \E \Big[ \big(\sum_{j \in W_t} a(j) u_t(j) \big)^2\Big] \leq O(1) \cdot \sum_{j \in W_t} a(j)^2 u_t(j)^2 ,
\end{align}
which intuitively means that the coordinates of $u_t$ is almost pairwise independent (up to a constant factor loss). 
Surprisingly, this seemingly weak form of independence implies strong subgaussian-type concentration for discrepancy in the algorithmic framework above \cite{BDG19,BG17}. In particular, if $u_t$ is $O(1)$-spectrally independent for all time $t$, then for any test vector $a \in \R^n$, its discrepancy change over any time period $[t_1, t_2]$ is essentially $O(\|a\|_2)$-subgaussian, i.e., 
\begin{align} \label{eq:spec-ind_subg}
\p(|\langle a, x_{t_2} - x_{t_1} \rangle| \geq c \cdot \|a\|_2) \leq \exp(- \Omega(c^2)) .
\end{align}

\smallskip
\noindent \textbf{The Bansal-Garg Bound.} The tail bound in \eqref{eq:spec-ind_subg} turns out to be sufficient to recover Bansal and Garg's results in \cite{BG17}. Here, we show their $O(d \sqrt{\log n})$ bound for $\ell_\infty$ to $\ell_2$ prefix discrepancy (their $O(\sqrt{d \log n})$ bound for $\ell_2$ to $\ell_2$ prefix discrepancy follows from a similar argument). 
Recall that Banaszczyk's non-constructive bound is $O(d + \sqrt{d \log n})$ in the $\ell_\infty$ to $\ell_2$ setting.

Fix a prefix $\mathcal{P} \in [n]$. Denote $A_i^{\mathcal{P}} \in [-1,1]^\mathcal{P}$ the row $A_i$ restricted to the prefix $\mathcal{P}$, and denote $\varphi_t^\mathcal{P}(i) = \langle A_i^{\mathcal{P}}, x_t \rangle$ the discrepancy of row $i$ for prefix $\mathcal{P}$. We will show that with probability $1 - 1/\poly(n)$, $|\varphi^{\mathcal{P}}_n(i)| \leq O(\sqrt{d \log n})$ for all row $i \in [d]$, which implies that $\|\varphi^{\mathcal{P}}_n\|_2 \leq O(d \sqrt{\log n})$. 
Taking a union bound over all prefixes $\mathcal{P} \in [n]$ then gives the Bansal-Garg result. 

To bound $\varphi_n^{\mathcal{P}}(i)$, let $t^*$ be the first time when column $\mathcal{P}$ enters the sliding window, and let $a_i^{\mathcal{P}} \in \R^{W_{t^*}}$ be the vector $A_i^{\mathcal{P}}$ restricted to that particular sliding window $W_{t^*}$. Clearly,  we have $\|a_i^{\mathcal{P}}\|_2 \leq \sqrt{10 d}$ as the sliding window only has $10d$ alive columns.   
Since prefix $\mathcal{P}$ does not incur any discrepancy before it enters the sliding window, we have $\varphi_{t^*}^{\mathcal{P}}(i) = 0$. 
Note that only the coloring update of columns in $W_{t^*}$ can contribute to $\varphi_t^{\mathcal{P}}$ after time $t^*$, and thus $\varphi_n^{\mathcal{P}}(i) = \langle a_i^{\mathcal{P}}, x_n - x_{t^*} \rangle$. Then using \eqref{eq:spec-ind_subg} with $c = O(\sqrt{\log n})$ for large enough constants and $\|a_i^{\mathcal{P}}\|_2 \leq O(\sqrt{d})$ gives the desired bound
\[
\p(|\varphi_n^{\mathcal{P}}(i)| \geq O(\sqrt{d\log n})) \leq \exp(- \Omega(c^2)) = 1/\poly(n) .
\]

\smallskip
\noindent \textbf{Where Can We Improve?} 
The main reason for the  $O(\sqrt{\log n})$ factor loss in the Bansal-Garg bound is the following. Although each $\varphi_n^{\mathcal{P}}(i)$ is $O(d)$-subgaussian, the different $\varphi_n^{\mathcal{P}}(i)$'s may have arbitrary correlations, so they need to bound each $\varphi_n^{\mathcal{P}}(i)$ up to a tail probability of $1/\poly(n)$ to take a union bound over all rows $i \in [d]$ and prefixes $\mathcal{P} \in [n]$. 
This {\em row-wise analysis} is generally the best possible without any additional control over the  correlations among the $\varphi_n^{\mathcal{P}}(i)$'s.

However, in an ideal scenario where the discrepancies of all rows $\varphi_n^{\mathcal{P}}(i)$ are evolving independently,  standard concentration inequalities (e.g., see \cite[Theorem 3.1.1]{Ver18book}) give the following stronger tail bound for $\|\varphi_n^{\mathcal{P}}\|_2$ (than the row-wise analysis above) with an {\em additive} deviation:
\begin{align*}
\p\big(\|\varphi_n^{\mathcal{P}}\|_2 \geq O(\sqrt{d}(\sqrt{d} + c))\big) \leq \exp(- \Omega(c^2)) .
\end{align*}
Taking $c = O(\sqrt{\log n})$ here allows for a union bound over all prefixes $\mathcal{P} \in [n]$, and this would already recover Banaszczyk's non-constructive $O(d + \sqrt{d \log n})$ result. 

Of course, as the discrepancies $\varphi_n^{\mathcal{P}}(i)$ of different rows $i \in [d]$ are generated by the same dynamics $d x_t$, it is not realistic to assume that they are independent. 
Nonetheless, this dependency among row discrepancies is exactly the issue encountered by Bansal and Jiang \cite{BJ25b,BJ26} when studying the Beck-Fiala and Koml\'os conjectures. To bypass the issue and achieve substantial progress towards these long-standing conjectures, they devised a new technique to ``decouple'' the discrepancies of different rows using a new set of SDP constraints for $u_t$ called {\em affine spectral independence}. 

\subsection{Affine Spectral Independence and ASI-Guarded Prefixes}
\label{subsec:ASI-overview}

The idea in \cite{BJ26} is  to add a new set of affine spectral independence (ASI) constraints to the SDP to ensure that $d \varphi_t^{\mathcal{P}} = \varphi_{t + dt}^{\mathcal{P}} - \varphi_t^{\mathcal{P}}$ is $\gamma_{\mathsf{ASI}}$-spectrally independent (see \Cref{sec:framework} for details), i.e., 
\begin{align} \label{eq:ASI-overview}
\Cov(d \varphi_t^{\mathcal{P}}) \preceq \gamma_{\mathsf{ASI}} \cdot \diag(\Cov(d \varphi_t^{\mathcal{P}})).  \tag{$\mathsf{ASI}$} 
\end{align}
Similar to \eqref{eq:pairwise-ind-overview}, this intuitively says that the discrepancy change $d \varphi_t^{\mathcal{P}}(i)$ of different rows $i \in [d]$ are almost pairwise independent up to a $\gamma_{\mathsf{ASI}}$ factor, and we show that this weaker form of independence suffices for an additive bound. 
In particular, we analyze the dynamics of $d \|\varphi_t^{\mathcal{P}}\|_2^2$ as in \cite{BG17}, 
\[
d \|\varphi_t^{\mathcal{P}}\|_2^2 = 2 \underbrace{\langle \varphi_t^{\mathcal{P}}, d \varphi_t^{\mathcal{P}} \rangle}_{dL^{\mathcal{P}}_{t}} + \underbrace{\langle d \varphi_t^{\mathcal{P}}, d \varphi_t^{\mathcal{P}} \rangle}_{dQ^{\mathcal{P}}_{t}} .
\]
\cite{BG17} already gave a bound for the quadratic term $Q_t^{\mathcal{P}}$ matching Banaszczyk's bound, using \eqref{eq:SI-overview} and a Freedman-type martingale analysis (see \Cref{subsecn:proofFramework}), and the main bottleneck there lies in controlling the linear term $L_t^{\mathcal{P}}$. 
It turns out that the \eqref{eq:ASI-overview} property reduces the {\em quadratic variation} of $L_t^{\mathcal{P}}$, which allows us to show that with probability $1 - 1/\poly(n)$, 
\begin{align} \label{eq:ell_infty-to-ell_2-overview}
\|\varphi_t^{\mathcal{P}}\|_2 \leq O(d +  \sqrt{\gamma_{\mathsf{ASI}} \cdot d \log n}). 
\end{align}
We defer the technical details to \Cref{sec:ell_inf-to-ell_2}, but this essentially recovers Banaszczyk's non-constructive bound if the ASI factor $\gamma_{\mathsf{ASI}}$ could be set to $O(1)$ as in the spectral independence property \eqref{eq:SI-overview}.

\smallskip
\noindent \textbf{Controlling All Prefixes Blows Up the ASI Factor.} Unfortunately, while we can indeed set $\gamma_{\mathsf{ASI}} = O(1)$ in \eqref{eq:ASI-overview} for a {\em fixed} prefix $\mathcal{P}$, a crucial issue arises when controlling multiple prefixes --- the $\gamma_{\mathsf{ASI}}$ factor suffers from the number of prefixes one needs to control simultaneously \cite{BJ26}. As the sliding window $W_t$ has size $10d$, naively satisfying \eqref{eq:ASI-overview} for all prefixes $\mathcal{P} \in W_t$ would result in setting $\gamma_{\mathsf{ASI}} = \Theta(d)$, for which the bound in \eqref{eq:ell_infty-to-ell_2-overview} is no better than the Bansal-Garg bound. 

To bypass this issue, we control a carefully chosen set of  $o(d)$ prefixes $\mathcal{I}_t \subseteq W_t$ that we call {\em ASI-guarded  prefixes} --- this allows us to choose $\gamma_{\mathsf{ASI}} =|\mathcal{I}_t| = o(d)$ (and we will think of $\gamma_{\ASI} = |\mathcal{I}_t|$ henceforth), for which the bound in \eqref{eq:ell_infty-to-ell_2-overview} improves upon Bansal-Garg. 
However, this forces us to lose the \eqref{eq:ASI-overview} control over {\em most} prefixes (i.e., those that are not ASI-guarded, which we will refer to as {\em unguarded prefixes}), and this poses two further technical challenges: 
\begin{enumerate}
    \item [(i)] As the sliding window $W_t$ evolves over time, we need to update the  set of $o(d)$ many ASI-guarded prefixes dynamically. 
    \item [(ii)]  We need to control the discrepancies of all the $10d - o(d)$ unguarded prefixes. In particular, some of the previous ASI-guarded prefixes might become unguarded later on, and we can no longer rely on \eqref{eq:ell_infty-to-ell_2-overview} to bound their discrepancies. 
\end{enumerate}

To get around these challenges, we design a data structure for maintaining the set of $o(d)$ ASI-guarded prefixes that we call {\em Global Interval Tree} (and denoted as $\GlobalIntervalTree$). 

\subsection{Global Interval Tree: Controlling Unguarded Prefixes} 
\label{subsec:IntTree-overview}

In this subsection, we give an overview of the $\GlobalIntervalTree$ data structure and discuss how it overcomes the above challenges (see Sections \ref{subsec:globalIntervalTree} and \ref{subsecn:proofFramework} for more details).

\smallskip
\noindent \textbf{The Interval Representation.} Denote $\mathcal{I}_t \subseteq W_t$ the set of $\gamma_{\ASI} = o(d)$ ASI-guarded prefixes that $\GlobalIntervalTree$ maintains at each time $t$. 
Note that $\mathcal{I}_t$ naturally corresponds to a collection of consecutive intervals in $[n]$, where each pair of adjacent ASI-guarded prefixes $\mathcal{P}_1, \mathcal{P}_2 \in \mathcal{I}_t$ induces the interval $(\mathcal{P}_1, \mathcal{P}_2] := \{\mathcal{P}_1 + 1, \mathcal{P}_1 + 2,  \cdots, \mathcal{P}_2\}$. 
This collection of consecutive intervals essentially forms a partition (except for a final interval between $\max\{\mathcal{I}_t\}$ and $\max\{W_t\}$) of the columns in $[W_t]$, where we use $[W_t] := [\min\{W_t\}, \max\{W_t\}]$ to denote the set of all columns, alive or dead, that lies inside the range of the current sliding window $W_t$. We refer to $[W_t]$ as the {\em complete sliding window}, and note that $W_t$ only contains all the alive columns in $[W_t]$. 
\begin{figure}[h]
    \centering
    \includegraphics[width=0.6\textwidth]{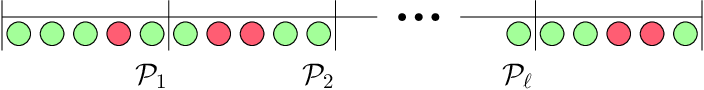}
    \caption{Partition of the complete sliding window $[W_t]$ by ASI-guarded prefixes $\mathcal{I}_t = \{\mathcal{P}_1, \cdots, \mathcal{P}_{\ell} \}$. Throughout, green dots denote alive columns, and red dots denote dead ones. \label{fig:TimeWindowPartition}}

\end{figure}

From this interval perspective, an equivalent task is to maintain a set  of $o(d)$ intervals that partition $[W_t]$, which we abuse notation and also denote as $\mathcal{I}_t$. 

As the sliding window $W_t$ progresses over time, the set of intervals $\mathcal{I}_t$ needs to be updated dynamically. For our purposes, we will perform two operations: 

(1) add new intervals into $\mathcal{I}_t$ when new columns enter $[W_t]$, and (2) merge consecutive intervals in $\mathcal{I}_t$ when their sizes, i.e., the number of alive columns they currently contain, become too small (and thus too expensive to maintain separately).\footnote{Note that merging two consecutive intervals corresponds to unguarding the ASI-guarded prefix between them.} 

\smallskip
\noindent \textbf{Adding and Merging Intervals.} To describe these operations formally, we partition $[n]$ into $n/s$ {\em base intervals} of size $s = \omega(1)$ each at time $0$, and use $\mathcal{B}$ to denote this collection of base intervals. We call a base interval $I \in \mathcal{B}$ {\em active} at time $t$ if $I \subseteq [1, \max\{W_t\}]$ and $I \cap W_t \neq \emptyset$, i.e., when $I$ contains at least one active column and no column in $I$ is dormant. 
At time $0$, the initial sliding window $W_0$ (as well as $\mathcal{I}_0$) contains the first $10d/s = o(d)$ base intervals. 
We add a base interval $I \in \mathcal{B}$ to $\mathcal{I}_t$ whenever it becomes active, and we will {\em only add base intervals to $\mathcal{I}_t$ throughout}.

An interval $I \in \mathcal{I}_t$ is said to be {\em small} if its size is at most $s/2$. 
Ideally, we would like to merge two consecutive intervals when both of them are small to ensure that there are no small intervals, but we will need to do this carefully --- as we will see in the following discussion, merging an interval too many times may blow up the discrepancy of the prefixes it contains.

\smallskip
\noindent \textbf{Two Conflicting Goals.} Before describing our merging strategy, we first address the challenges mentioned towards the end of \Cref{subsec:ASI-overview}. 
For challenge (i), the total number of intervals in $\mathcal{I}_t$ must be $o(d)$, which is equivalent to the average size of intervals in $\mathcal{I}_t$ being $\omega(1)$: 
\begin{align*}
   \text{ Goal (i): the average interval size in $\mathcal{I}_t$ should be $\omega(1)$. }
\end{align*}
To control the discrepancy of an unguarded prefix $\mathcal{P} \in [W_t]$ as posted in challenge (ii), let $I \in \mathcal{I}_t$ be the interval containing $\mathcal{P}$. The left endpoint of $I$ is an ASI-guarded prefix which we will call the {\em ASI-guard} for $\mathcal{P}$ at time $t$.
From this perspective, the discrepancy change for prefix $\mathcal{P}$ can be decomposed as $d \varphi_t^{\mathcal{P}} = d \varphi_t^{\mathsf{ASI}} + d \varphi_t^{\mathsf{err}}$, where $d \varphi_t^{\mathsf{ASI}}$ is the discrepancy change of its ASI-guard, and $d \varphi_t^{\mathsf{err}}$ is the discrepancy change of the (alive) columns between the ASI-guard and $\mathcal{P}$, which we will refer to as the {\em error columns} or the {\em error set} (see \Cref{fig:disc_decomposition} for an illustration). 

\medskip

\begin{figure}[h]
    \centering
    \includegraphics[width=0.6\textwidth]{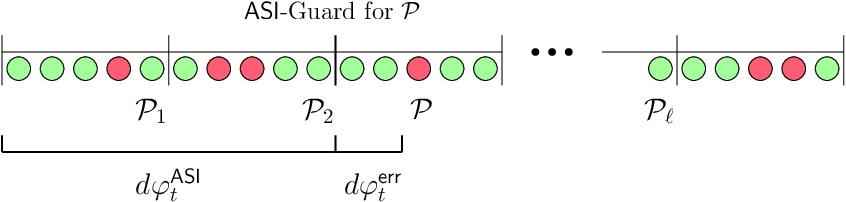}
    \caption{\label{fig:disc_decomposition} ASI-guard for prefix $\mathcal{P}$, and the corresponding ASI and error components of $d \varphi_t^{\mathcal{P}}$.}
\end{figure}

\smallskip

Roughly speaking, we will show (see \Cref{subsecn:proofFramework} for details) that 
the contribution from the ASI-guard part $d \varphi_t^{\mathsf{ASI}}$ can still be bounded as in \eqref{eq:ell_infty-to-ell_2-overview}, even though the ASI-guard for $\mathcal{P}$ might change over time --- thus as long as goal (i) is achieved, one has $\gamma_{\ASI} = |\mathcal{I}_t| = o(d)$ and this suffices for bounding the ASI-guard part $\varphi_t^{\ASI}$ better than the Bansal-Garg bound.

However, the discrepancy contribution of the error part $d \varphi_t^{\mathsf{err}}$ depends on the total number of columns that have ever been part of the error set for prefix $\mathcal{P}$ at any time --- it can be bounded by $O(\sqrt{(\text{total \# of error columns}) \cdot \log n})$ using the Freedman-type martingale analysis in \cite{BG17}. 
Note that the total number of error columns for prefix $\mathcal{P}$ depends on how many times the interval $I$ containing $\mathcal{P}$ merges with the interval on its left (see Figure \ref{fig:growthOfErrorSet}), and can be upper bounded by
\[
(\text{total \# of times interval containing $\mathcal{P}$ merges}) \cdot (s/2) ,
\]
where recall that $s/2$ is the threshold size for a small interval. This leads to the following goal:  

\begin{align*}
    \text{Goal (ii): for every prefix $\mathcal{P}$, the number of times the interval containing it merges is small. }
\end{align*}
Note that there is an obvious tension between the two goals above --- avoiding small-sized intervals as in goal (i) naturally leads to substantial number of merges --- but we will be able to achieve both goals using the strategy of merging via a global binary tree. 
\begin{figure}[h]
    \centering
    \includegraphics[width=0.45\textwidth]{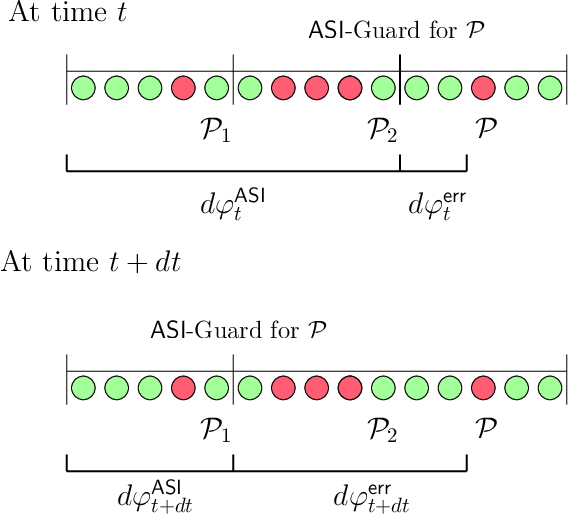}
    \caption{\label{fig:growthOfErrorSet} The error set grows after the interval containing $\mathcal{P}$ merges with the previous interval.}
\end{figure}

\smallskip
\noindent \textbf{Merging via a Global Binary Tree.} 
At time $0$, we view the $n/s$ base intervals as the leaves of a complete global binary tree $\mathsf{Tree}_0$, whose depth is $\log(n/s) =  O(\log n)$ (here we pretend that $n/s$ is a power of $2$ for simplicity). 
At each time $t$, we maintain $\mathsf{Tree}_t$ to be a subtree of $\mathsf{Tree}_0$, such that the leaf nodes of $\mathsf{Tree}_t$ (restricted to $[W_t]$) correspond to all the intervals maintained in $\mathcal{I}_t$.

Roughly speaking, when two base intervals corresponding to sibling leaves in $\mathsf{Tree}_0$ are both small, we will merge them into a single interval represented by their parent (and delete the two sibling leaves from $\mathsf{Tree}_0$). 
More generally, our rule for merging intervals is the following:\footnote{For technical reasons, the actual merging rule that we use is slightly more sophisticated (see \Cref{mergerule:global_binary_tree}).}
\begin{quote}
{\bf Merging Rule:} Whenever two consecutive intervals $I_1, I_2 \in \mathcal{I}_t$ are both small at time $t$ and they correspond to the two sibling children of a parent node $I_0 \in \mathsf{Tree}_t$, we will merge $I_1$ and $I_2$ into an interval for $I_0$, and delete the two sibling children from $\mathsf{Tree}_t$. 
\end{quote}

We will show that this merging rule ensures that each node only merges at most $O(\log n)$ times as the depth of $\mathsf{Tree}_0$ is $O(\log n)$, while at the same time, the average interval size is at least $\Omega(s/\log n)$. 
Setting $s =\omega(1)$ properly then achieves our two goals above. 
See \Cref{subsec:globalIntervalTree} for details.

\section{Preliminaries and Our Framework}
\label{sec:framework}

In this section, we formally describe our algorithmic framework and the meta-analysis that we use to achieve our results.
We start with some notation. 

\smallskip
\noindent \textbf{Notation.} 
As in \Cref{subsec:nutshell}, we view the vectors $v_1, \cdots, v_n \in \R^d$ as a matrix $A \in \R^{d \times n}$ whose $j$th column is $v_j$, and use $A_i \in \R^n$ to denote its $i$th row. 
For any $W \subseteq [n]$, denote $A_i(W) \in \R^W$ the row $A_i$ restricted to columns in $W$; for a vector $u \in \R^W$, we sometimes abuse notation and view it as a vector in $\R^n$ with all-zero coordinates in $[n] \setminus W$. 
For any prefix $\mathcal{P} \in [n]$, define $A^{\mathcal{P}} \in \R^{d \times \mathcal{P}}$ (resp. $A_i^{\mathcal{P}} \in \R^{\mathcal{P}}$) as the matrix $A$ (resp. row $A_i$) restricted to the columns in prefix $\mathcal{P}$.  

We use $\varphi$ for discrepancy, and $\varphi^{\mathcal{P}}_t$ to denote the discrepancy vector for prefix $\mathcal{P}$ at time $t$.  
We always use $i \in [d]$ to index rows of $A$, and $j \in [n]$ to index its columns.

Throughout,  $v(\ell)$ is the $\ell$th entry of vector $v$ and $\log$ means logarithm to the base $2$. 

We write $\lesssim$ (resp. $\gtrsim$) when the left-hand side is at most (resp. at least) a constant times the right-hand side.

We define a notion of intervals for subsets of $[n]$. 
An interval in $[n]$ is a consecutive set of elements in $[n]$. For $\ell,r \in [n]$ with $\ell \leq r$, we denote $[\ell, r] := \{\ell, \cdots, r \}$, $(\ell, r] := \{\ell+ 1, \cdots, r\}$, $[\ell, r) = \{\ell, \cdots, r-1\}$ the closed, left-open, and right-open intervals with boundaries $\ell$ and $r$. 
For a set $W \subseteq [n]$, we denote $[W] := [\min_{\ell \in W} \{\ell\}, \max_{\ell \in W} \{\ell\}]$ the smallest interval containing $W$.

\subsection{Algorithmic Framework: the SDP and Coloring Update}\label{subsecn:updateFramework}

As described in \Cref{subsec:basic_framework}, our algorithms find the desired coloring by evolving a fractional coloring $x_t \in [-1,1]^n$, with $x_0 = 0^n$, over time $t$ that increments in tiny steps of size $d t= 1/\poly(n)$. 
The fractional coloring $x_t$ is updated as $d x_t =  u_t \sqrt{dt}$, where $u_t \in \R^{W_t}$ is a carefully-chosen random unit vector supported on the sliding window $W_t \subseteq [n]$ consisting of the first $10d$ alive columns,\footnote{\label{footnote:Wt-final}If there are less than $10d$ alive columns left, the sliding window $W_t$ will contain all the remaining alive columns. This will happen after $W_t$ hits the last column of the input matrix $A$, and in this case, we will not enforce the constraint $\sum_{j \in W_t} A_i(j) u_t(j) = 0$ for any row $i \in [d]$. None of our analysis will be affected by this change of the size $|W_t|$ at the end of the algorithm, and we will ignore this nuance throughout for simplicity.} 

so that $|W_t| = 10d$. In particular, $u_t$ will be chosen so that $u_t \perp x_t$ and that the discrepancy change of the whole sliding window $W_t$ is $0$, i.e., $\sum_{j \in W_t} A_i(j) u_t(j) = 0$ for all rows $i \in [d]$.

As mentioned earlier, to choose $u_t$ at each time $t$, our algorithms first computes a PSD matrix $U_t \in \R^{W_t \times W_t}$ by solving an SDP, and then samples $u_t$ with $U_t$ as its covariance matrix (up to a constant scaling factor so that $\|u_t\|_2 = 1$). 
We first describe this SDP below and then explain the intuition behind its constraints. 
In prior works, this SDP has a few more parameters, but here we hard-code most of them to be constants to simplify our presentation. 

\subsubsection{The SDP and Its Feasibility} 
\label{subsubsec:SDP_feasibility}
Let the factor $\gamma_{\ASI} = o(d)$ in \eqref{eq:ASI-overview} be a parameter, to be specified later (depending on the problem setting) but fixed throughout the algorithm. 
As mentioned in \Cref{subsec:ASI-overview}, our algorithm maintains a set of at most $\gamma_{\ASI}$ many ASI-guarded prefixes $\mathcal{I}_t$ using the $\GlobalIntervalTree$ data structure.\footnote{The upper bound of $|\mathcal{I}_t| \leq \gamma_{\ASI}$ is to ensure that the SDP remains feasible at every time step $t \geq 0$. } We discussed it briefly in \Cref{subsec:IntTree-overview} and will give more details in \Cref{subsec:globalIntervalTree}.

Based on the $\GlobalIntervalTree$ data structure, at each time $t$, the algorithm chooses a subspace $H_t \subseteq \R^{W_t}$ with $\dim(H_t) \leq 0.1 |W_t| \leq d$, and a matrix $E_t \in \R^{(0.1\gamma_{\ASI} |W_t|) \times W_t}$. The rows of $E_t$ exactly corresponds to the $d$ rows of $A$ restricted to the at most $ 0.1 \gamma_{\ASI} |W_t|/d \leq \gamma_{\ASI}$ many ASI-guarded prefixes $\mathcal{I}_t$.

Consider the following SDP with matrix variable $U \in \R^{W_t \times W_t}$ and parameter $\gamma_{\ASI}$.

\begin{subequations}\label{eq:FrameworkSDP}
\begin{align}
U_{j,j} &\leq 1 
\quad &&\text{for all } j \in W_{t}, \label{eq:sdp:diag} \\
\operatorname{Tr}(U) &\geq 0.1 \cdot |W_t|, \label{eq:sdp:trace} \\
\langle U , A_i(W_t) (A_i(W_t))^\top\rangle &= 0 
\quad &&\text{for all } 1 \leq i \leq d, \label{eq:sdp:block-rows} \\
\langle U, ww^\top \rangle &= 0 
\quad &&\text{for all } w \in H_t, \label{eq:sdp:block-W} \\
U &\preceq 10 \cdot \operatorname{diag}(U),  \label{eq:sdp:SI} \tag{4e SI} \\
E_{t} U E_t^{\top} &\preceq \gamma_{\ASI} \cdot \operatorname{diag} ( E_t U E_t^{\top}) , \tag{4f ASI} \label{eq:sdp:ASI}  \\
U &\succeq 0 . \tag{4g} \label{eq:sdp:psd}
\end{align}
\end{subequations}

To understand SDP \eqref{eq:FrameworkSDP}, consider a feasible solution $U_t$ and a mean-zero random coloring update $u_t \in \R^{W_t}$ with $U_t \propto \Cov(u_t)$ (up to a scaling factor).
We will show how to sample $u_t$ given the PSD matrix $U_t$ in \Cref{subsubsec:sampling_vt}. Here, we first explain what the SDP constraints imply for $u_t$. 

\smallskip
\noindent \textbf{Spreadness, Sliding Window, and Blocking Constraints.} 
The first four constraints are standard and easy to parse. 
Constraints \eqref{eq:sdp:diag} and \eqref{eq:sdp:trace} ensures that the coloring update $u_t$ ``uniformly spreads out'' among its coordinates, and avoids the trivial solution $U = 0$. 
Constraint \eqref{eq:sdp:block-rows} is equivalent to $\langle A_i(W_t), u_t \rangle = \sum_{j\in W_t} A_i(j) u_t(j) = 0$, and this ensures that the whole sliding window has $0$ discrepancy update,\footnote{As mentioned in \Cref{footnote:Wt-final}, this sliding window constraint will be dropped when $W_t$ hits the last input vector.\label{footnote:drop_sliding_window_constrarint}} as mentioned earlier. 
Constraint \eqref{eq:sdp:block-W} enforces that $u_t \perp H_t$, and this can be understood as {\em blocking} any vector $w \in H_t$ so that its discrepancy does not change, i.e., $\langle w, u_t \rangle = 0$. In particular, our algorithm always chooses $x_t \in H_t$ to ensure that $u_t \perp x_t$, as promised earlier. 

\smallskip
\noindent \textbf{Spectral Independence and Affine Spectral Independence.} 
Corresponding to \eqref{eq:SI-overview}, the spectral independence constraint \eqref{eq:sdp:SI} guarantees that the random coloring update vector $u_t$ is $10$-spectrally independent. As explained in \Cref{subsec:SI-BGbound-Overview}, this intuitively says that the coordinates of $u_t$ is almost pairwise independent up to a constant factor. 

Analogously, corresponding to \eqref{eq:ASI-overview}, the affine spectral independence constraint \eqref{eq:sdp:ASI} 
guarantees that the discrepancy update $d \varphi_t^{\mathcal{P}}$ is $\gamma_{\ASI}$-spectrally independent for all ASI-guarded prefixes $\mathcal{P} \in \mathcal{I}_t$. Intuitively, this says that the coordinates of $d \varphi_t^{\mathcal{P}}$ are almost pairwise independent up to a $\gamma_{\ASI}$ factor. 
Note that the ASI factor $\gamma_{\ASI}$ is proportional to the number of ASI-guarded prefixes in $\mathcal{I}_t$, and this dependency is shown to be unavoidable in general \cite{BJ26}.

Using standard duality arguments, the feasibility of the SDP \eqref{eq:FrameworkSDP} was given in \cite{BJ26}. 

\begin{fact}[SDP feasibility, \cite{BJ26}] \label{fact:sdpFeasibility}
For any subspace $H_t \subseteq \R^{W_t}$ with $\dim(H_t) \leq 0.1|W_t|$ and matrix $E_t \in \R^{(0.1 \gamma_{\ASI} |W_t|) \times W_t}$, the SDP \eqref{eq:FrameworkSDP} given by constraints \eqref{eq:sdp:diag}-\eqref{eq:sdp:psd} is feasible. 
\end{fact}

In fact, \Cref{fact:sdpFeasibility} is a special case of Theorem A.4 from \cite{BJ26}. 
For completeness, we include the statement of Theorem A.4 from \cite{BJ26} and explain how it implies \Cref{fact:sdpFeasibility} in \Cref{appendix:SDPFeasibility}.

\subsubsection{Sampling $u_t$ from an SDP Solution}
\label{subsubsec:sampling_vt}

Once a  feasible SDP solution $U_t \in \R^{W_t \times W_t}$ is computed, the random coloring update vector $u_t$ with $\Cov(u_t) \propto U_t$ can be sampled using  the following standard approach (e.g., \cite{BJ26}).  
In particular, let $U_t = Q_t \Lambda_t Q_t^\top$ be its spectral decomposition. We choose 
\begin{align} \label{eq:find_vt}
    u_t := (\Tr(U_t))^{-1/2} \, U_t^{1/2} Q_t r_t = (\Tr(U_t))^{-1/2}\, Q_t \Lambda_t^{1/2} r_t,
\end{align}
where $r_t \in \R^{W_t}$ is a random vector with i.i.d.~Rademacher random variables (taking values $1$ or $-1$ with probability $1/2$ each). Note that this ensures $\E[u_t] = 0$, and also $\|u_t\|_2=1$ as  
\begin{align*} 
\|u_t\|_2^2 = u_t^\top u_t = \frac{1}{\Tr(U_t)} r_t^\top \Lambda_t^{1/2} Q_t^\top Q_t \Lambda_t^{1/2} r_t = \frac{\Tr(\Lambda_t)}{\Tr(U_t)} = 1 .
\end{align*}

\subsection{Global Interval Tree}
\label{subsec:globalIntervalTree}
As mentioned in \Cref{subsecn:updateFramework}, the subspace $H_{t}$ in the blocking constraint \eqref{eq:sdp:block-rows} and the matrix $E_{t}$ in the ASI constraint \eqref{eq:sdp:ASI} are chosen based on the $\mathcal{T}_{\mathsf{Global}}$ data structure. In this section, we formally describe this data structure and state its properties that we will need.

To simplify our presentation here, we only describe the $\mathcal{T}_{\mathsf{Global}}$ data structure in the context of choosing the set of $ \gamma_{\mathsf{ASI}}$ many ASI-guarded prefixes $\mathcal{I}_t \subseteq W_{t}$, corresponding to the rows of $E_t$ in the ASI constraint \eqref{eq:sdp:ASI}. The choice of the subspace $H_t$ for the blocking constraint \eqref{eq:sdp:block-rows} requires a slight modification of the data structure given here. But as it is only needed for the proof of the $\ell_2$ to $\ell_2$ prefix discrepancy in \Cref{thm:main-ell2-to-ell2}, we will defer the details to \Cref{sec:ell2-to-ell2}.

In \Cref{subsubsec:intervalRepresentationAndOperation}, we describe how $\mathcal{T}_{\mathsf{Global}}$ maintains and updates a partition of the complete sliding window $[W_t]$ over time,  
following the discussion from Section \ref{subsec:IntTree-overview}. Then, in Section \ref{subsubsec:average-interval-size} and Section \ref{subsubsec:prefixErrorSet}, we prove the properties of $\mathcal{T}_{\mathsf{Global}}$ that we need.

\subsubsection{Interval Representation and Operations} \label{subsubsec:intervalRepresentationAndOperation}

Recall from Section \ref{subsec:IntTree-overview} that a set of ASI-guarded prefixes $\mathcal{I}_{t} \subseteq W_t$ can be equivalently viewed as
a collection of consecutive intervals that partition the complete sliding window  $[W_{t}]$.
Formally, $\mathcal{I}_t = \{\mathcal{P}_1, \mathcal{P}_2, \ldots, \mathcal{P}_\ell\}$ splits  $[W_t]$ into consecutive intervals $(\mathcal{P}_j, \mathcal{P}_{j+1}]$, for $j = 0, \cdots, \ell-1$ and $\mathcal{P}_0 := \min\{W_t \}$, 
along with a final interval $(\mathcal{P}_\ell, \max\{W_t\}]$ containing all columns in $[W_t]$ larger than $\mathcal{P}_\ell$ (see Figure~\ref{fig:TimeWindowPartition}). 
We abuse notation and also use $\mathcal{I}_{t}$ to denote the set of intervals in this partition of $[W_{t}]$ (we will ignore the final interval $(\mathcal{P}_\ell, \max\{W_t\}]$) when the context is clear.

The data structure maintains a set $\mathcal{I}_t$ of at most $\gamma_{\mathsf{ASI}}$ intervals that partitions $[W_t]$ at each step.
Recall from \Cref{subsec:IntTree-overview} that as the sliding window $W_t$ progresses, the data structure updates the intervals in $\mathcal{I}_t$ via two operations: (1) add new intervals to $\mathcal{I}_t$ containing newly activated columns in $[W_t]$, or (2) merge consecutive intervals already in $\mathcal{I}_t$ when their sizes (i.e., number of alive columns contained in them) become too small. We describe these operations formally below. 

\medskip
\noindent \textbf{Initialization and Adding Intervals.} Before the process starts, we partition $[n]$ into $n/s$ {\em base intervals} of size $s = 20\gamma_{\mathsf{ASI}}^{-1}\,d \log n$ each, and use $\mathcal{B}$ to denote this collection of base intervals. 
At time $0$, we initialize $\mathcal{I}_0$ to be the first $|W_0|/s=  10d/s$ base intervals contained in $W_0$. 
At each time $t$, a base interval $I \in \mathcal{B}$ is called {\em active} if $I \subseteq [1, \max\{W_t\}]$ and $I \cap W_t \neq \emptyset$ (i.e., $I$ contains at least one active column and no column in $I$ is dormant).
We add a base interval $I \in \mathcal{B}$ to $\mathcal{I}_t$ whenever it becomes active, and we will only add base intervals to $\mathcal{I}_t$ throughout.

\begin{remark} \label{remark:WindowEndConstraint}
To ensure that columns in the final interval $(\mathcal{P}_t, \max\{W_t\}]$ are also controlled by the \eqref{eq:sdp:ASI} constraint, the data structure always includes $\max\{W_t\}$ as a ``special'' ASI-guarded prefix in $\mathcal{I}_t$. This ensures that the entire sliding window $W_t$ also satisfies \eqref{eq:sdp:ASI}, and doesn't affect the SDP feasibility. 
To simplify our presentation, we will not make this  extra ASI constraint explicit. 

\end{remark}

\medskip
\noindent \textbf{Merging via a Global Binary Tree.} 

The intervals in $\mathcal{I}_t$ are merged 
 
using a global binary tree, denoted by $\mathsf{Tree}_t$, that evolves over time. 
Initially, $\mathsf{Tree}_0$ is the complete binary tree of height $\log (n/s)$ whose leaves correspond to all the base intervals $\mathcal{B}$ (see Figure~\ref{fig:BinaryTree}), and at each time $t$, tree $\mathsf{Tree}_t$ is a sub-tree of $\mathsf{Tree}_0$. 

\begin{figure}[htbp]
    \centering
    \includegraphics[width=0.8\textwidth]{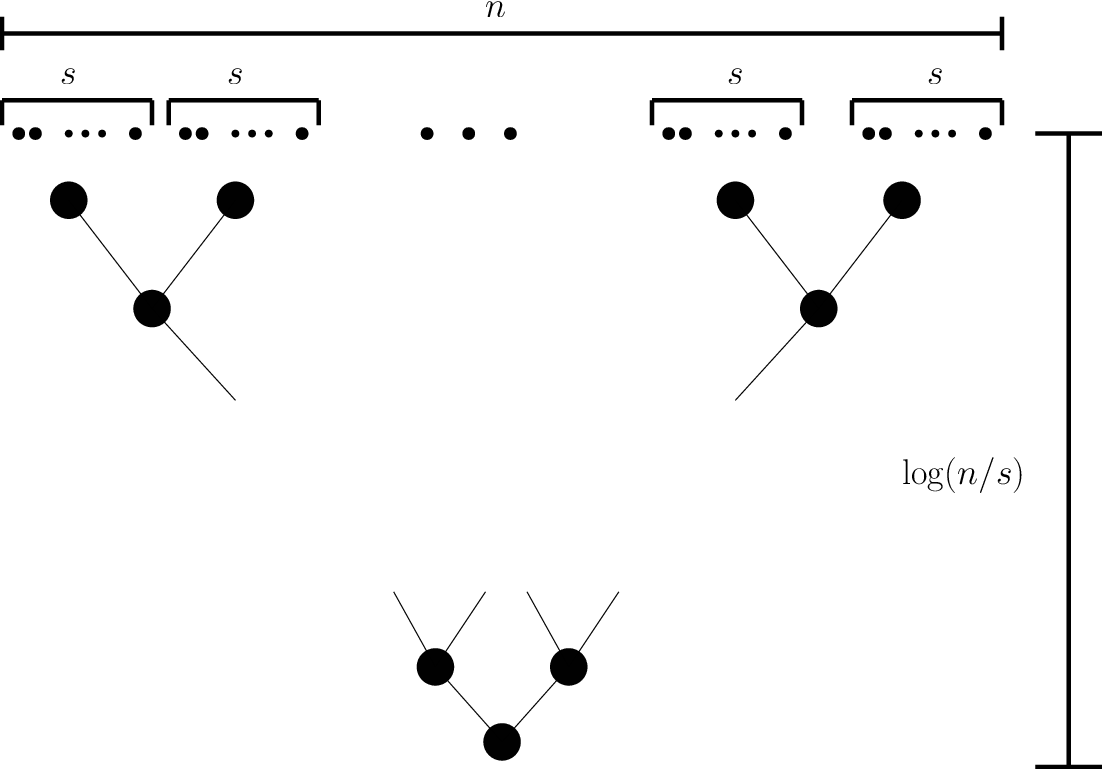}
    \caption{$\mathsf{Tree}_0$ is a complete binary tree whose leaves correspond to all base intervals $\mathcal{B}$.}
    \label{fig:BinaryTree}
\end{figure}

Note that the leaves of $\mathsf{Tree}_0$ correspond exactly to the first $10d/s$ base  intervals $\mathcal{I}_0$ when restricted to the sliding window $W_0$, and the rest of the leaves to base intervals not yet added to $\mathcal{I}_0$.
We maintain this property throughout --- namely, we ensure
that the leaves of $\mathsf{Tree}_t$ correspond either to  intervals in $\mathcal{I}_t$, or to base intervals in $\mathcal{B}$ that have not been added to $\mathcal{I}_t$ yet. 
We call the former type of leaves of $\mathsf{Tree}_t$ {\em active}, and refer to them and the intervals in $\mathcal{I}_t$ interchangeably.

For each active leaf $v \in \mathsf{Tree}_t$, we use $\mathsf{size}(v)$ to denote the number of alive columns contained in the interval of $\mathcal{I}_t$ corresponding to it, and call $v$ {\em small} if $\mathsf{size}(v) < s/2$. 
We refer to the two children of each node in $\mathsf{Tree}_t$ as {\em left} and {\em right} child as they appear in Figure~\ref{fig:BinaryTree}. 
A leaf is called a {\em left leaf} (resp. {\em right leaf}) if it is the left (resp. right) child of its parent. 
At any time $t$, we update $\mathsf{Tree}_t$ by deleting certain active leaves (or equivalently, merging the corresponding interval in $\mathcal{I}_{t}$ with one of its neighboring intervals). 
Formally, our merging rule is the following.\footnote{It may seem natural to use the simple rule of merging two sibling leaves whenever they are both small, but this might create a sequence of $\Theta(\log n)$ small left leaves (corresponding to a root-leaf path) if the sibling right leaf at the end of this sequence is not yet active. While this sequence of small left leaves is fine for the purpose of maintaining the ASI-guarded prefixes $\mathcal{I}_t$, it will create some technical issues when one tries to adapt the data structure to maintain the subspace $H_t$ in \Cref{sec:ell2-to-ell2}. This is why we use the more sophisticated \Cref{mergerule:global_binary_tree} instead.}

\begin{mergerule}[Merging via global binary tree] \label{mergerule:global_binary_tree}
At any time $t$: 
\begin{enumerate}
    \item [(i)] For any small left leaf $v \in \mathsf{Tree}_t$, delete it and merge its interval with the interval at the next active leaf (see Figure \ref{fig:mergeTree} for an illustration); do nothing if $v$ is the last active leaf.\footnote{\label{footnote:last-active-left-leaf}The case where an active small left leaf $v$ is the last active leaf in $\mathsf{Tree}_t$ will be ignored throughout our analysis. One can essentially view $v$ as being deleted and merged with the next base interval when it becomes active.} In particular, if the right sibling of $v$ is also an active leaf, then also delete the right sibling and use their parent to represent the merged interval (see \Cref{fig:rightChildDelete}). 
    \item [(ii)] For any small right leaf $v \in \mathsf{Tree}_t$, if the left sibling of $v$ has been deleted, then delete $v$ and use its parent to represent its interval (see \Cref{fig:rightChildDeleteCaseII} for an illustratoin).
\end{enumerate}
\end{mergerule}

\begin{figure}[htbp]
    \centering
    \includegraphics[width=0.8\textwidth]{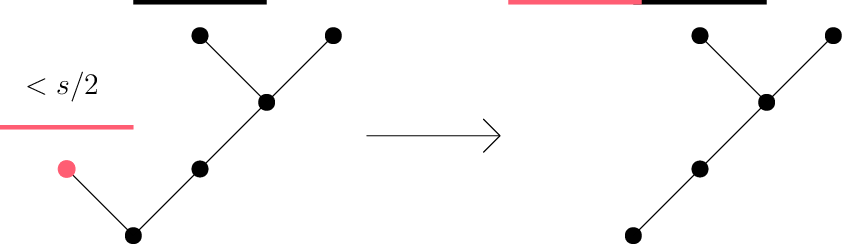}
    \caption{When a left leaf becomes small, it is deleted and merged with the next active leaf.}
    \label{fig:mergeTree}
\end{figure}

\begin{figure}[htbp]
    \centering
    \includegraphics[width=0.7\textwidth]{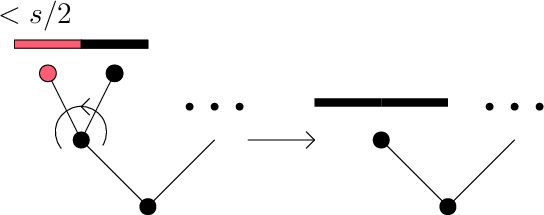}
    \caption{When the right sibling of a small left leaf is also an active leaf, both leaves get deleted and the parent represents the merged interval.}
    \label{fig:rightChildDelete}
\end{figure}

\begin{figure}[htbp]
    \centering
    \includegraphics[width=0.7\textwidth]{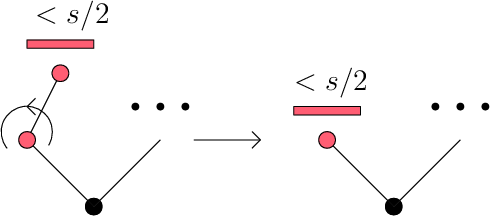}
    \caption{When the left sibling of a small right leaf has been deleted, the small right left gets deleted and the parent represents its interval. Note that the parent is now a small left leaf, and it will be merged subsequently according to \Cref{mergerule:global_binary_tree} (i), as illustrated in \Cref{fig:mergeTree}.}
    \label{fig:rightChildDeleteCaseII}
\end{figure}

Intuitively, one can view the merging by \Cref{mergerule:global_binary_tree} as happening from left to right --- as (i) describes, whenever a left active leaf becomes small, it immediately merges with the active interval on its right if there exists one.
The deletion in (ii) is mainly for the purpose of cleaning up $\mathsf{Tree}_t$ when only a single right child of certain nodes in the tree survives.

\subsubsection{Bounding the Number of Intervals}
\label{subsubsec:average-interval-size}

\Cref{mergerule:global_binary_tree} allows us to lower bound the average size of active leaves in $\mathsf{Tree}_t$, and hence to upper bound their number, which is equal to the number of ASI-guarded prefixes $|\mathcal{I}_t|$.

\begin{proposition}[Bounding the number of intervals]\label{prop:boundASIConstraints}
At any time $t$, consider $\mathsf{Tree}_t$ after \Cref{mergerule:global_binary_tree} is completed. Then the average size of active leaves in $\mathsf{Tree}_t$ is at least $s/(2 \log n)$. Consequently, the total number of active leaves is   $|\mathcal{I}_{t}| \leqslant (2|W_t| \log n)/s \leq\gamma_{\mathsf{ASI}}$. 
\end{proposition}
\begin{proof}
Note that there is no small left leaves in $\mathsf{Tree}_t$ after \Cref{mergerule:global_binary_tree} is completed (except for possibly the last active leaf, which we will ignore as mentioned in \Cref{footnote:last-active-left-leaf}), as they all get merged by (i). 
There might be many small right leaves in $\mathsf{Tree}_t$, but we will control their average sizes by charging them to certain non-small active left leaves as follows.

For each small right leaf $v$, consider its left sibling $w$, which must exist as otherwise $v$ will be deleted by (ii). 
Note that in the subtree of $\mathsf{Tree}_t$ rooted at $w$ (which may contain the single node $w$ if it is an active leaf), there must exist an active non-small left leaf $w'$, as otherwise the subtree rooted at $w$ will get merged by (i). We will charge $v$ to this non-small left leaf $w'$. 

Note that each small right leaf that gets charged to $w'$ is the sibling of one of its ancestors. As $w'$ can have $<\log n$ ancestors (corresponding to the entire root-leaf path), $<\log n$ small right leaves can be charged to each non-small left leaf $w'$. 
As $w'$ has size at least $s/2$, averaging over itself and all the $<\log n$ small right leaves that are charged to it, the average sizes of active leaves in $\mathsf{Tree}_t$ is at least $s/(2 \log n)$.
The upper bound on $|\mathcal{I}_t|$ then follows immediately, as the intervals corresponding to the active leaves of $\mathsf{Tree}_t$ form a partition of the sliding window $W_t$. 

\end{proof}

\subsubsection{Bounding the Error Sets for Prefixes}
\label{subsubsec:prefixErrorSet}

Recall from \Cref{subsec:IntTree-overview} (also see \Cref{subsecn:proofFramework}) that for any $\mathcal{P} \in [W_t]$ that is not an ASI-guarded prefix in $\mathcal{I}_t$, we call it {\em unguarded}, and assign $\mathcal{P}_t := \max \{\mathcal{I}_t \cap [\mathcal{P}]\}$, i.e., the maximum-indexed ASI-guarded prefix before $\mathcal{P}$, as its ASI-guard at time $t$ (see Figure \ref{fig:disc_decomposition}). 
For an ASI-guarded prefix $\mathcal{P} \in \mathcal{I}_t$, its ASI-guard at time $t$ is defined to be itself. 
 
Recall that the discrepancy change for prefix $\mathcal{P}$ is decomposed as $d \varphi_t^{\mathcal{P}} = d \varphi_t^{\mathsf{ASI}} + d \varphi_t^{\mathsf{err}}$, where $d \varphi_t^{\mathsf{ASI}}$ is the discrepancy change of its ASI-guard $\mathcal{P}_t$, and $d \varphi_t^{\mathsf{err}}$ is the that of the alive columns between the ASI-guard and $\mathcal{P}$. 
We refer to this latter set the {\em error} columns for $\mathcal{P}$ at time $t$. 
The union of error columns for $\mathcal{P}$ at all time steps is called its {\em error set}, and is formally defined below.

\begin{definition}[Error set for $\mathcal{P}$] \label{defn:error-set}
The error set for $\mathcal{P}$ up till time $t$, denoted as $\mathsf{Error}^{\mathcal{P}}_t$, consists of all error columns for $\mathcal{P}$ for all time steps $t' \leq t$. 
The error set for $\mathcal{P}$ is defined as $\mathsf{Error}^{\mathcal{P}} := \mathsf{Error}^{\mathcal{P}}_n$. 
\end{definition}

Note that the error set for $\mathcal{P}$ depends on the outcomes of the randomness of the algorithm.
As discussed in \Cref{subsec:IntTree-overview}, our bound on $\varphi_t^{\err}$ will depend on how large the error set for $\mathcal{P}$ is. 
\Cref{mergerule:global_binary_tree} allows us to give the following upper bound on $\Error^{\mathcal{P}}$.

\begin{proposition}[Bounding the error set]\label{prop:boundASIErrorSet}
We always have $|\mathsf{Error}^{\mathcal{P}}|\leq (s/2) \log n = O(\gamma_{\mathsf{ASI}}^{-1}\,d\log^{2}n)$.
\end{proposition}

\begin{proof}
Our goal is to find a set $|\mathcal{G}_\mathcal{P}| < \log n$ of nodes in $\mathsf{Tree}_0$ that may contribute to $\Error^{\mathcal{P}}$, and show that the subtree rooted at each node in $\mathcal{G}_\mathcal{P}$ can contribute at most $s/2$ to $\Error^{\mathcal{P}}$. 

To define the node set $\mathcal{G}_\mathcal{P}$, let $v_\mathcal{P} \in \mathcal{B}$ be the leaf of $\mathsf{Tree}_0$ (equivalently, base interval) that contains $\mathcal{P}$. 
Consider the root-leaf path from the root $\mathsf{root}_0$ of $\mathsf{Tree}_0$ to $v_\mathcal{P}$. 
Starting from $\mathsf{root}_0$ and following along this root-leaf path, whenever a node $v$ on the root-leaf path is the right child of its parent, we add the left sibling of $v$ to $\mathcal{G}_\mathcal{P}$ (see Figure \ref{fig:prefixPartition} for an illustration).

As the depth of $\mathsf{Tree}_0$, and hence the length of the root-leaf path, is $\log (n/s)$, we have $\ell := |\mathcal{G}_\mathcal{P}| < (\log n) - 1$.  
Denote the nodes in $\mathcal{G}_{\mathcal{P}}$ as $G_{\mathcal{P}}(1), \cdots, G_{\mathcal{P}}(\ell)$.
We abuse notation and also use $G_{\mathcal{P}}(k)$ to denote the union of all leaf intervals in the subtree rooted at node $G_{\mathcal{P}}(k)$.
Note that the intervals $G_{\mathcal{P}}(1), \cdots, G_{\mathcal{P}}(\ell)$, together with the (partial) base interval $v_{\mathcal{P}} \cap [\mathcal{P}]$, form a partition of $[\mathcal{P}]$, and these are the only intervals that can contribute to $\Error^{\mathcal{P}}$. 
We bound their contributions below.

For each $k \in [\ell]$, note that as long as the node $G_\mathcal{P}(k)$ is not deleted, none of the intervals in its subtree will merge with the interval containing $\mathcal{P}$, and thus contribute no column to $\mathsf{Error}^{\mathcal{P}}$. 
Only when $G_\mathcal{P}(k)$ becomes an active (left) leaf that is small can it merge with the interval containing $\mathcal{P}$, but in this case $\mathsf{size}(G_\mathcal{P}(k)) < s/2$ and thus it contributes at most $s/2$ columns to $\Error^{\mathcal{P}}$. 
As the partial base interval $v_{\mathcal{P}} \cap [\mathcal{P}]$ can contribute at most $\mathsf{size}(v_{\mathcal{P}}) \leq s$ columns to $\Error^{\mathcal{P}}$, the size of $\Error^{\mathcal{P}}$ is bounded by $(s/2) \ell + s \leq (s/2) \log n$, and this completes the proof. 

\end{proof}

\begin{figure}[htbp]
    \centering
    \includegraphics[width=0.6\textwidth]{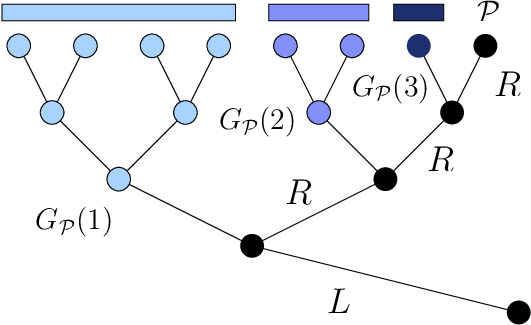}
    \caption{Decomposition of prefix $\mathcal{P}$ based on its root-leaf path. Each node $G_{\mathcal{P}}(k)$ added to $\mathcal{G}_{\mathcal{P}}$ is the left sibling of a node on the root-leaf path that is the right child of its parent's.}
    \label{fig:prefixPartition}
\end{figure}

\subsection{A Meta-Analysis}
\label{subsecn:proofFramework}

In this subsection, we give a meta-analysis of our algorithms that is common to the proofs of Theorems \ref{thm:main-ell2-to-ell2} and \ref{thm:ell_inf-to-ell_2}. 
The specific instantiations of this meta-analysis and the detailed proofs in these two different settings will be given in the next two sections.

Fix an arbitrary prefix $\mathcal{P} \in [n]$. Our goal is to obtain a bound on the $\ell_2$ norm of its discrepancy vector $\varphi_t^{\mathcal{P}}$ (with entries $\varphi_t^{\mathcal{P}}(i) = \langle A_i^{\mathcal{P}},x_t\rangle$) with probability $1 - 1/\poly(n)$, so that we can union bound over all prefixes. 
Let $t^*$ be the first time step when $\mathcal{P} \in W_{t^*}$. 
As our algorithm ensures that the discrepancy update for the whole sliding window is $0$ by SDP constraint \eqref{eq:sdp:block-W}, $\varphi_{t^*}^{\mathcal{P}} = 0$ and the discrepancy vector $\varphi_t^{\mathcal{P}}$ is entirely due to the coloring update of the alive columns in $W_{t^*}$ after time $t^*$.  
In the following, we will condition on $W_{t^*}$ and analyze the discrepancy update after $t^*$.
To keep our notation simple, we drop the superscript $\mathcal{P}$ whenever the context is clear.

As mentioned in \Cref{subsec:ASI-overview}, we will analyze the dynamics of $\|\varphi_t\|_2^2$ given by
\begin{align}\label{eq:IncrementNormDecomposition}
d\lVert \varphi_{t} \rVert_{2}^{2} = \lVert \varphi_{t + dt} \rVert_{2}^{2} - \lVert \varphi_{t} \rVert _{2}^{2} = 2\underbrace{\langle d\varphi_{t}, \varphi_{t}\rangle}_{dL_{t}} + \underbrace{\langle d\varphi_{t}, d\varphi_{t}\rangle}_{dQ_{t}},
\end{align}
and bound the processes $L_t$ and $Q_t$ (with increments $d L_t$ and $d Q_t$) separately. 
It turns out that just using spectral independence \eqref{eq:sdp:SI}, \cite{BG17} already gave a bound on $Q_t$ that matches Banaszczyk's non-constructive bound for the (more general) $\ell_2$ to $\ell_2$ prefix discrepancy setting. 

\begin{restatable}[$Q_t$ bound for $\ell_2$ to $\ell_2$ prefix discrepancy, \cite{BG17}]{fact}{QuadraticTerm} \label{fact:quadratic_term}
Consider the setting of \Cref{thm:main-ell2-to-ell2}. For the algorithm in \Cref{subsecn:updateFramework}, one has $Q_t \leq O(d + \log n)$ with probability $1 - 1/\poly(n)$. 
\end{restatable}

The above statement follows from Claim~27 in the proof of Theorem~24 of \cite{BG17}.

\smallskip
\noindent \textbf{Bounding the Linear Term $L_t$.} 
The bottleneck for the analysis in \cite{BG17}, which results in their sub-optimal bound, is the linear term $L_t$. To apply Freedman's inequality (see \Cref{fact:freedman}) for the martingale $L_t$, they need to bound its quadratic variation $\E_t[(d L_t)^2]$, where $\E_t$ takes the expectation conditional on the events on before time $t$. However, as they have no control over the vector $d \varphi_t$, they need to apply the Cauchy-Schwartz inequality
\begin{align} \label{eq:linear_CauchySchwartz_BG17}
\E_t\big[(d L_t)^2 \big] \leq \E_t \big [\|d \varphi_t\|_2^2\big]  \cdot \|\varphi_t\|_2^2 .
\end{align}
Note that \eqref{eq:linear_CauchySchwartz_BG17} is quite wasteful when $d \varphi_t$ has non-trivial randomness, and our key improvement over \cite{BG17} comes from the randomness guaranteed by affine spectral independence \eqref{eq:sdp:ASI} for the set of ASI-guarded prefixes in $\mathcal{I}_t$. 

However, as the prefix $\mathcal{P}$ may not be an ASI-guarded prefix, we don't have almost pairwise independence for $d \varphi_t$ and thus cannot  directly improve over \eqref{eq:linear_CauchySchwartz_BG17}. 
Instead, recall from \Cref{subsubsec:prefixErrorSet} that we let $\mathcal{P}_t := \max \{\mathcal{I}_t \cap [\mathcal{P}]\}$ be the ASI-guard for prefix $\mathcal{P}$ at time $t$, and decompose $d \varphi_t$ as 
\begin{align}\label{eq:IncrementDiscrepancyVectorDecomposition}
d\varphi_t = d\varphi_t^{\ASI} + d\varphi_t^{\err},
\end{align}
where $d\varphi_t^{\ASI} := d \varphi_t^{\mathcal{P}_t}$ is the discrepancy change of the ASI-guard $\mathcal{P}_t$, and $d \varphi_t^{\err}$ is the discrepancy change due to the error columns between $\mathcal{P}_t$ and $\mathcal{P}$ (see \Cref{subsubsec:prefixErrorSet}). 
Correspondingly, 
\begin{align}\label{eq:LtDecomposition}
d L_t= \langle d\varphi_t^{\ASI}, \varphi_t \rangle + \langle d\varphi_t^{\err}, \varphi_t\rangle =: d L_t^{\ASI}  + d L_t^{\err} .
\end{align}

In our analysis, we will control the two terms $d L_t^{\ASI}$ and $d L_t^{\err}$ separately.

Roughly speaking, the ASI-guard part $d L_t^{\ASI}$ can be bounded better than \eqref{eq:linear_CauchySchwartz_BG17} because   $d\varphi_t^{\ASI}$ always satisfies affine spectral independence \eqref{eq:sdp:ASI}, despite that the ASI-guard $\mathcal{P}_t$ might change over time.  
To control the error part $d L_t^{\err}$, note that the $\GlobalIntervalTree$ data structure guarantees that there will be at most $O(s \log n) = O(\gamma_{\ASI}^{-1} \log^2 n)$ error columns throughout all time steps $t$, which allows for a standard Freedman-type analysis.  
We defer the details of how to bound these two parts of the linear term, depending on the specific problem assumptions, to subsequent sections.

\section{$\ell_\infty$ to $\ell_2$ Prefix Discrepancy}\label{sec:ell_inf-to-ell_2}

In this section, we prove \Cref{thm:ell_inf-to-ell_2}, which is restated below for convenience.

\ellinftoelltwo*

We present our algorithm for \Cref{thm:ell_inf-to-ell_2} in \Cref{subsec:ell_inf-to-ell_twoALG} and its analysis in \Cref{subsec:ell_inf-to-ell_twoAnalysis}. 
The proof of \Cref{thm:ell_inf-to-ell_2} will appear in \Cref{subsec:ellinf_to_ell2_putting_together}.

\subsection{Algorithm}\label{subsec:ell_inf-to-ell_twoALG}
Fix $\gamma_{\mathsf{ASI}} := 100\sqrt{d}\log n$ and $\tau := \Theta(d + d^{3/4} \log n + d^{1/4} \log^{3/2} n)$ the target discrepancy bound in \Cref{thm:ell_inf-to-ell_2} (with a large enough constant). The algorithm follows the framework in \Cref{subsec:basic_framework} and it also uses the $\GlobalIntervalTree$ data structure as described in \Cref{subsec:globalIntervalTree} to maintain a set of at most $\gamma_{\ASI}$ ASI-guarded prefixes $\mathcal{I}_t \subseteq W_t$. 

At each time step $t$, the algorithm does the following.

\begin{enumerate}
\item If there exists prefix $\mathcal{P} \in [n]$ whose $\lVert \varphi_t^{\mathcal{P}} \rVert_2$ exceeds the target $\tau$, it outputs $\mathsf{ABORT}$. 

\item Otherwise, it solves the SDP in \eqref{eq:FrameworkSDP}, where $H_t = \emptyset$ (i.e., there is no constraint of the form \eqref{eq:sdp:block-rows}) and the matrix $E_t \in \R^{d |\mathcal{I}_t| \times W_t}$ contains a row $A_i^{\mathcal{P}}(W_t)$ (viewed as a vector in $\R^{W_t}$ with coordinates $0$ in $W_t \setminus [\mathcal{P}]$) for every ASI-guarded prefix $\mathcal{P} \in \mathcal{I}_t$ and $i \in [d]$. 

\end{enumerate}

By \Cref{prop:boundASIConstraints} and the setting of $s = 20 \gamma_{\ASI}^{-1} d \log n$ in \Cref{subsubsec:average-interval-size}, the row dimension of $E_t$ is 
\[
d |\mathcal{I}_t| \leq  d  \cdot \frac{2|W_t| \log n}{s} = \frac{2 d |W_t| \log n}{20 \gamma_{\ASI}^{-1} d \log n}= 0.1 \gamma_{\ASI} |W_t| ,
\]

which satisfies the condition of \Cref{fact:sdpFeasibility}. Thus the SDP \eqref{eq:FrameworkSDP} is feasible at every time step $t$, and to bound the $\ell_2$ prefix discrepancy of this algorithm and prove \Cref{thm:ell_inf-to-ell_2}, it suffices to prove that the algorithm does not $\mathsf{ABORT}$ with high probability.

\subsection{Analysis}
\label{subsec:ell_inf-to-ell_twoAnalysis}

Fix any prefix $\mathcal{P} \in [n]$. We will show that, with probability $1 - 1/\mathsf{poly}(n)$, the $\ell_2$ prefix discrepancy $\lVert \varphi_{t}^{\mathcal{P}}\rVert_{2}$ is at most $\tau$ at every time step $t$. Since the algorithm runs for $\mathsf{poly}(n)$ steps and there are $n$ prefixes, a union bound across all time steps and prefixes completes the proof. In the following, we drop the superscript $\mathcal{P}$ whenever it is clear from the context.

\smallskip
\noindent \textbf{Road Map for the Analysis.}
Following the meta-analysis in \Cref{subsecn:proofFramework}, at each time step $t$, we decompose the change of the squared $\ell_2$ prefix discrepancy $d \|\varphi_t\|_2^2$ as 
\[
d \|\varphi_t\|_2^2 = d Q_t + 2 ( d L_t^{\ASI} + d L_t^{\err}).
\]
Our goal will be to show that for any time $t \in [0,n]$, with probability $1 - 1/\poly(n)$, 
\begin{align} \label{eq:stopping_condition}
Q_t \leq \tau^2/3 \ , \quad L_{t}^{\ASI} \leq \tau^2/6 \ \text{ and } \ \  L_{t}^{\err} \leq \tau^2/6  .
\end{align}
Once the above is shown, then by a union bound, one has $\|\varphi_t\|_2^2 \leq \tau^2$ for all time steps $t$, which implies that the algorithm doesn't $\mathsf{ABORT}$ with high probability. 

For technical reasons, in our analysis, we will use the more stringent stopping condition (than the algorithm $\mathsf{ABORT}$s) based on \eqref{eq:stopping_condition}. Namely, if any of the conditions in \eqref{eq:stopping_condition} is violated at any time $t$, we will freeze the process $x_t$ and thus $L_t^{\ASI}, L_t^{\err}$ and $Q_t$ will have zero increment onward.

\begin{abort}
\label{rule:ell_inf_to_ell_twoSTOP}
    We freeze the process $x_t$ onward if any of the three conditions in \eqref{eq:stopping_condition}, i.e., $Q_t \leq \tau^2/3$, $L_t^{\mathsf{ASI}} \leq \tau^2/6$, and $L_t^{\mathsf{err}} \leq \tau^2/6$, is violated. 
\end{abort}

Note that the modified process given by \Cref{rule:ell_inf_to_ell_twoSTOP} violates \eqref{eq:stopping_condition} at some time step if and only if the unmodified process $x_t$ does (at a possibly different time). Because of this equivalence, we will abuse notation and also refer to the modified process as $x_t$, and the quadratic and linear terms for the modified process as $Q_t$, $L_t^{\ASI}$, and $L_t^{\err}$. 

For the modified process, we always have 
\begin{align} \label{eq:modified_process_bound}
\|\varphi_t\|_2^2 \leq Q_t + L_t^{\ASI} + L_t^{\err} \leq \tau^2/3 + 3 (\tau^2/6 + \tau^2/6) = \tau^2 .
\end{align}
We will argue that with probability at most $1/\poly(n)$, the modified process violates \eqref{eq:stopping_condition}, and we will do so by taking a union bound over the probability that each one of the three conditions in \eqref{eq:stopping_condition} gets violated. 
As discussed in \Cref{subsecn:proofFramework}, the bound of $Q_t \leq O(d^2 + d \log n)\leq \tau^2/3$ with probability $1 - 1/\poly(n)$ was already shown in \cite{BG17} (see \Cref{fact:quadratic_term}). 

In Section \ref{subsubsec:boundingLinearASI} and Section \ref{subsubsec:boundingLinearErr} below, we bound $L_{t}^{\mathsf{ASI}}$ and $L_{t}^{\mathsf{err}}$ in Lemmas \ref{lem:bound_ell_infASI} and \ref{lem:bound_ellinf_error} respectively.

\subsubsection{Bounding the ASI-Guard Part $L_{t}^{\mathsf{ASI}}$}\label{subsubsec:boundingLinearASI}

We first bound the ASI-guard term $L_{t}^{\mathsf{ASI}}$ in the following lemma.

\begin{lemma}[Bounding the ASI-guard part, $\ell_\infty$ to $\ell_2$]\label{lem:bound_ell_infASI}
    Consider the algorithm in Section \ref{subsec:ell_inf-to-ell_twoALG}. For any time $t$, with probability $1 - 1/\poly(n)$, 
    \[
    L_{t}^{\mathsf{ASI}} \leq \tau^{2}/6.
    \]
\end{lemma}
Our proof relies on the following Freedman-type  concentration inequality for super-martingales. 

\begin{fact}[Lemma 2.2 in \cite{Ban24}]  \label{fact:freedman}
Let $\{Z_t: t = 0,1, \cdots\}$ be a sequence of random variables with increments $\Delta Z_t := Z_t - Z_{t-1}$, such that $Z_0$ is deterministic and $\Delta Z_t \leq M$ for all $t \geq 1$.

If for all $t \geq 1$, we have 
\begin{align} \label{eq:freedman_drift_vs_variation}
\E_{t-1}[\Delta Z_t ] \leq - \delta\, \E_{t-1}[(\Delta Z_t)^2]
\end{align}
holds with $0 < \delta < 1/M$, where $\E_{t-1}[\cdot]$ denotes $\E[\cdot | Z_1, \cdots, Z_{t-1}]$. Then for all $\xi \geq 0$, 
\[
\p\big( Z_t - Z_0 > \xi \big) \leq \exp(- \delta \xi). \]
\end{fact}

\begin{proof}[Proof of \Cref{lem:bound_ell_infASI}]
Fix a time step $t_0 \geq 0$, we bound the probability that $L_{t_0}^{\mathsf{ASI}}$ exceeds $\tau^{2}/6$.

 Recall that $\varphi_t$ gets its first non-zero increment after the first time $t^{\star}$ when the sliding window $W_{t^*}$ contains $\mathcal{P}$. We may assume that $t_0 \geq t^*$, as otherwise $L_{t_0}^{\ASI}$ will be $0$.  Following \cite{BG17}, we denote $\mathsf{Corr} = \mathsf{Corr}^{\mathcal{P}}:= W_{t^*} \cap [\mathcal{P}]$ the columns in $W_{t^*}$ with index at most $\mathcal{P}$. Note that only columns in $\mathsf{Corr}$ contribute non-zero discrepancy to $\varphi_{t}$.
    Similar to the analysis in \cite{BG17,Ban24,BJ26}, we define the following regularized ASI-guard term
    \begin{align*}
Y_{t} := L_{t}^{\mathsf{ASI}} - \beta_{\mathsf{ASI}}\sum_{j\in \mathsf{Corr}}x_{t}^{2}(j), 
    \end{align*}
where we set $\beta_{\mathsf{ASI}} = \Theta(\gamma_{\mathsf{ASI}}\log n)$ for a sufficiently large constant. 
Below, we will show that the increment $dY_t$ satisfies condition \eqref{eq:freedman_drift_vs_variation} with a suitably chosen factor $\delta$. Here, we naturally map time steps $t \in [0,n]$ to $\{0,1, \cdots, n/d t\}$.

If at any time $t$, the modified process $x_t$ is frozen due to the condition \eqref{eq:stopping_condition} being violated prior to $t$, then $L_t^{\mathsf{ASI}}$, and hence $Y_t$, must have been zero increment at time $t$. Thus the increment $d Y_t = 0$ trivially satisfies the condition \eqref{eq:freedman_drift_vs_variation} in \Cref{fact:freedman} for any $\delta$ in this case.

In the other case where the modified process $x_t$ has not been frozen up till time $t$. We will show that condition \eqref{eq:freedman_drift_vs_variation} in \Cref{fact:freedman} holds for some factor $\delta > 0$ by computing the first and second moments of $d Y_t$ in the following.

Note that the increment of $Y_{t}$ is
\begin{align*}
    dY_{t} = \langle d\varphi_{t}^{\mathsf{ASI}}, \varphi_{t}\rangle -2\beta_{\mathsf{ASI}}\sum_{j \in \mathsf{Corr}}u_{t}(j)x_{t}(j)\sqrt{dt}  - \beta_{\mathsf{ASI}}\sum_{j \in \mathsf{Corr}}u_{t}(j)^2 \,dt.
\end{align*}

In what follows, all expectations are conditioned on the filtration $\mathcal{F}_t$, or outcome of randomness, up to time $t$, and is denoted as $\mathbb{E}_t[\cdot] = \mathbb{E}[\cdot \mid \mathcal{F}_t]$. 
Since $\mathbb{E}_t[u_{t}] = 0$,
the first moment of $dY_{t}$ is
\begin{align}
\mathbb{E}_t[dY_{t}] = -\beta_{\mathsf{ASI}}\,\sum_{j \in \mathsf{Corr}}\mathbb{E}_t[u_{t}(j)^2] dt\label{eq:LinearASIFirstMoment}.
\end{align}

The second moment of $d Y_t$ is bounded as
\begin{align*}
\mathbb{E}_t[(dY_{t})^{2}] &\lesssim \mathbb{E}_t\langle d\varphi_{t}^{\mathsf{ASI}}, \varphi_{t}\rangle^{2} + \beta_{\mathsf{ASI}}^{2}\cdot \mathbb{E}_t\bigl(\sum_{j \in \mathsf{Corr}}u_{t}(j)x_{t}(j)\bigr)^{2}dt \\
&\lesssim \mathbb{E}_t\langle d\varphi_{t}^{\mathsf{ASI}}, \varphi_{t}\rangle^{2} + \beta_{\mathsf{ASI}}^{2}\,\sum_{j \in \mathsf{Corr}}\mathbb{E}_t[u_{t}^{2}(j)]dt,
\end{align*}
where we drop lower order terms with scale $O((dt)^{3/2})$, use $(a+b)^2 \leqslant 2a^2 + 2b^2$ in the first inequality, and use $O(1)$-spectral independence of $u_t$ and that $|x_t(j)| \leq 1$ in the second. We bound the first term using affine spectral independence: 
\begin{align*}
\mathbb{E}_t\langle d\varphi_{t}^{\mathsf{ASI}}, \varphi_{t}\rangle^{2} 
&\leq \gamma_{\mathsf{ASI}}\,\sum_{i \in [d]}\varphi_{t}(i)^2 \, \E_t\big[ d\varphi_{t}^{\mathsf{ASI}}(i)^2 \big]\\
&= \gamma_{\mathsf{ASI}}\cdot\sum_{i \in [d]}\varphi_{t}(i)^2 \cdot \E_t \Big(\sum_{j \in \mathsf{Corr}}A_{i}(j) u_{t}(j)\Big)^{2} dt \lesssim \gamma_{\mathsf{ASI}}\,\lVert \varphi_{t}\rVert_{2}^{2}\sum_{j \in \mathsf{Corr}}\mathbb{E}_t[u_{t}^{2}(j)] dt,
\end{align*}
where the first step follows via $\gamma_{\mathsf{ASI}}$-affine spectral independence of $u_{t}$, and the final step follows by $O(1)$-spectral independence of $u_{t}$ and using $|A_{i}(j)| \leq 1$ for all $i \in [d],j \in [n]$.

Combining, we get the following bound on the second moment of $d Y_t$, 
\begin{align}
    \mathbb{E}_t[(dY_{t})^{2}] \lesssim (\gamma_{\mathsf{ASI}}\,\lVert \varphi_{t}\rVert_{2}^{2} + \beta_{\mathsf{ASI}}^{2})\sum_{j \in \mathsf{Corr} }\mathbb{E}_t[u_{t}^{2}(j)]\,dt \leq (\gamma_{\mathsf{ASI}}\,\tau^{2} + \beta_{\mathsf{ASI}}^{2})\sum_{j \in \mathsf{Corr}}\mathbb{E}_t[u_{t}^{2}(j)]\,dt,\label{eq:LinearASISecondMoment}
\end{align}
where we bound $\lVert \varphi_{t}\rVert_{2}^{2}$ by $\tau^{2}$ using \eqref{eq:modified_process_bound}.

Combining \eqref{eq:LinearASIFirstMoment} and \eqref{eq:LinearASISecondMoment}, we get condition \eqref{eq:freedman_drift_vs_variation} with 
\[
\mathbb{E}_{t}[dY_{t}] \lesssim -\Big(\frac{\beta_{\mathsf{ASI}}}{\tau^{2}\gamma_{\mathsf{ASI}} + \beta_{\mathsf{ASI}}^{2}}\Big)\cdot \mathbb{E}_{t}[(dY_{t})^{2}]. 
\]
Then applying Freedman-type inequality (\Cref{fact:freedman}) with deviation $\xi = \tau^2/10$, we obtain 
\begin{align} \label{eq:ellinf_to_ell2_ASI_tail}
    \Pr(Y_{t_0} - Y_{t^*} \geq \xi) \leq \exp\Big(\frac{-\xi \,\beta_{\mathsf{ASI}}}{\beta_{\mathsf{ASI}}^{2} + \gamma_{\mathsf{ASI}}\,\tau^{2}} \Big) .
\end{align}
To bound the exponent in \eqref{eq:ellinf_to_ell2_ASI_tail}, note that since $\xi = \tau^2/10 = \Omega(\sqrt{d} \log^3 n)$ and $\beta_{\ASI} = \Theta(\gamma_{\ASI} \log n) = \Theta(\sqrt{d} \log^2 n)$ (as $\gamma_{\ASI} = 100 \sqrt{d }\log n$), the first term is bounded as $(\xi \beta_{\ASI})/\beta_{\ASI}^2 = \Omega(\log n)$. 
The second term can be bounded as 
\[
\frac{\xi \beta_{\ASI}}{\gamma_{\ASI} \tau^2} = \frac{(\tau^2/10) \cdot \Theta(\gamma_{\ASI} \log n)}{\gamma_{\ASI} \tau^2} = \Theta(\log n) .
\]
Thus as long as the constant in $\beta_{\ASI}$ is big enough and the constant in $\xi$ is even (much) larger, the tail probability in \eqref{eq:ellinf_to_ell2_ASI_tail} is at most $1/\poly(n)$. 

Finally, to compare the target event of $L_{t_0}^{\ASI} \geq \tau^2/6$ with the event $Y_{t_0} - Y_{t^*} \geq \xi$ in \eqref{eq:ellinf_to_ell2_ASI_tail}, note that 
\begin{align*}
Y_{t_0} - Y_{t^*} 
& = L_{t_0}^{\ASI} - \beta_{\ASI} \sum_{j \in \mathsf{Corr}} x_{t_0}(j)^2 - \Big(0 -  \beta_{\ASI} \sum_{j \in \mathsf{Corr}} x_{t^*}(j)^2\Big) \\
& = L_{t_0}^{\ASI}- \beta_{\ASI} \sum_{j \in \mathsf{Corr}} (x_{t_0}(j)^2 - x_{t^*}(j)^2) 
\ \geq \ L_{t_0}^{\ASI} - \beta_{\ASI} (10d)  \ \geq \  L_{t_0}^{\mathsf{ASI}} - \tau^2/15 ,
\end{align*}
where we used that $\beta_{\ASI} d = \Theta(d^{3/2} \log^2 n)$ and $\tau^2 = \Omega(d^{3/2} \log^2 n)$ (with a much larger constant). 
Consequently, the event $L_{t_0}^{\ASI} \geq \tau^2 / 6$ implies that $Y_{t_0} - Y_{t^*} \geq \tau^2/10$, and thus,
\[
\Pr \big( L_{t_0}^{\ASI} \geq \tau^2 / 6 \big) \leq \Pr(Y_{t_0} - Y_{t^*} \geq \xi) \leq 1/\poly(n) .
\]
Taking a union bound over all time steps $t_0$ completes the proof of the lemma. 

\end{proof}

\subsubsection{Bounding the Error Part $L_{t}^{\mathsf{err}}$}\label{subsubsec:boundingLinearErr}

Next, we show how to bound the error term $L_t^{\err}$ in the following lemma. 

\begin{lemma}[Bounding the error part, $\ell_\infty$ to $\ell_2$] \label{lem:bound_ellinf_error}
Consider the algorithm in Section \ref{subsec:ell_inf-to-ell_twoALG}. With probability $1 -1/\poly(n)$, we have 
\[
L_t^{\err} \leq \tau^{2}/6.
\]
\end{lemma}

\begin{proof}[Proof of \Cref{lem:bound_ellinf_error}]

Fix $t_{0} \geq 0$, we bound the probability that $L_{t_{0}}^{\mathsf{err}}$ exceeds $\tau^{2}/6$. Again we may assume that $t_0 \geq t^*$, where $t^*$ is the time step when prefix $\mathcal{P}$ first enters the time window.

Recall from \Cref{subsec:globalIntervalTree} that $ \Error^{\mathcal{P}}$ (resp. $\Error_t^{\mathcal{P}}$) is the set of columns that contribute to the error term $L_t^{\err}$ throughout the algorithm (resp. up till time $t$).
The set $\Error_t$, as a (random) function of $t$, is monotonically non-decreasing with several ``jumps'', corresponding to the times when the interval containing $\mathcal{P}$ merges in the $\mathcal{T}_{\mathsf{Global}}$ data structure. 

We drop the superscript $\mathcal{P}$ below. 

We define the following regularized process for $L_{t}^{\mathsf{err}}$, 
\begin{align*}
    Z_t := L_t^{\err} - \beta_{\err} \sum_{j \in \Error_t} x_t(j)^2,
\end{align*}
where $\beta_{\mathsf{err}} := \Theta(d\log n)$ for a large enough constant.
Note that whenever $\Error_t$ has a ``jump'', 
$Z_t$ decreases which is only helpful for us (and thus can be safely ignored henceforth). 

Similar to the proof of \Cref{lem:bound_ell_infASI}, the condition \eqref{eq:freedman_drift_vs_variation} is trivially satisfied if the modified process was frozen prior to time $t$. 
If the process was not frozen, the increment of $Z_t$ is given by
\begin{align*}
    d Z_t = \langle d\varphi_{t}^{\err}, \varphi_{t}\rangle - 2 \beta_{\err} \sum_{j \in \Error_t} u_t(j) x_t(j) \sqrt{dt} -  \beta_{\err} \sum_{j \in \Error_t} u_t(j)^2 d t . 
\end{align*}
Since both $d \varphi_t^{\err}$ and $u_t$ are mean-zero, the first moment of $d Z_t$ is:
\begin{align}
\mathbb{E}_t [dZ_{t} ] = -  \beta_{\err}\, \sum_{j \in \Error_t} \E_t [u_t(j)^2] d t. 
    \label{eq:linear_term-error}
\end{align}
The second moment can be bounded as
\begin{align*}
\mathbb{E}_t \big[ (dZ_{t})^2 \big] 
&= \mathbb{E}_t \Big(\langle d\varphi_{t}^{\err},\varphi_{t}\rangle - 2 \beta_{\err} \sum_{j \in \Error_t} u_t(j) x_t(j) \sqrt{dt} \Big)^{2} \nonumber \\
& \lesssim \mathbb{E}_t \big[\langle d\varphi_{t}^{\err},\varphi_{t}\rangle^{2} \big] + \beta_{\err}^{2} \cdot \E_t \Big(\sum_{j \in \Error_t} u_t(j) x_t(j) \Big)^2 dt \nonumber \\
& \lesssim \E_t \big[\|d\varphi_{t}^{\err}\|_2^2 \big] \cdot \|\varphi_t\|_2^2 + \beta_{\err}^{2} \cdot \sum_{j \in \Error_t} \E[u_t(j)^2]  dt,  
\end{align*}
where in the last inequality, we apply Cauchy-Schwartz inequality for the first term, and for the second term, we use that the vector $u_t$ is $O(1)$-spectrally independent and that each $|x_t(j)| \leq 1$. To bound $\E_t \big[\|d\varphi_{t}^{\err}\|_2^2\bigr]$ in the first term,  
we note that $d \varphi_t^{\err}(i) = \sum_{j \in \Error_t} A_i(j) u_t(j) \sqrt{d t}$. Then using that $u_t$ is $O(1)$-spectrally independent and each $|A_i(j)| \leq 1$, we can bound 
\begin{align*}
\E_t \big[\|d\varphi_{t}^{\err}\|_2^2 \big] 
& = \sum_{i \in [d]} \E_t \big[ (d\varphi_{t}^{\err}(i)) ^2 \big]    
\lesssim \sum_{i \in [d]} \sum_{j \in \Error_t} A_i(j)^2 \cdot \E_t[u_t(j)^2] dt \leq d\sum_{j \in \Error_t} \, \E_t[u_t(j)^2] dt.
\end{align*}
Plugging this into the above gives the bound
\begin{align} \label{eq:quadratic_term-error}
\mathbb{E}_t \big[ (dZ_{t})^2 \big] \leq \big(d \|\varphi_t\|_2^2 + \beta_\err^2 \big) \cdot \sum_{j \in \Error_t} \E_t[u_t(j)^2 ] dt.
\end{align}
Combining \eqref{eq:linear_term-error} and \eqref{eq:quadratic_term-error} gives us condition \eqref{eq:freedman_drift_vs_variation} with  
\[
\E_t[dZ_{t}] \leq - \Big(\frac{\beta_{\mathsf{err}}}{\tau^{2}d+ \beta_\err^2 }\Big) \cdot \E_t[(d Z_t)^2]. 
\]
Applying Freedman-type inequality (\Cref{fact:freedman}) for deviation $\xi = \tau^{2}/10$, we get 
\begin{align}
\Pr(Z_{t_0} - Z_{t^*} \geq \xi) \leq \exp \Big(\frac{-\beta_{\mathsf{err}}\,\xi}{\beta_{\mathsf{err}}^{2} + \tau^{2}d} \Big) \leq 1/\poly(n) ,
    \label{eq:TailBoundLinearASI}
\end{align}
where the last inequality is obtained by using $\beta_{\err} = \Theta(d \log n)$ (and thus $(\beta_{\err} \xi)/(\tau^2 d) = \Omega(\log n)$), and that $\xi = \tau^2/10 = \Omega(d^{3/2} \log^2 n)$ (which implies that $(\beta_{\err} \xi)/\beta_{\err}^2 = \Omega(d^{1/2} \log n)$). 

Finally, we note that the event $L_{t_0}^{\err} \geq \tau^2/6$ implies the event $Z_{t_0} - Z_{t^*} \geq \tau^2/10$, because
\begin{align*}
  Z_{t_0} - Z_{t^*} 
  & = L_{t_0}^{\err} - \beta_{\err} \sum_{j \in \Error_{t_0}} \big(x_{t_0}(j)^2 - x_{t^*}(j)^2 \big) \geq L_{t_0}^{\err} - \beta_{\err} \cdot |\Error| \\
  & \geq L_{t_0}^{\err} - \beta_{\err} \cdot O(\gamma_{\ASI}^{-1} d \log^2 n) = L_{t_0}^{\err} - O(1) \cdot \frac{d^2 \log^3 n}{\sqrt{d} \log n}  \geq L_{t_0}^{\err} - \tau^2/15 ,
\end{align*}
where the first inequality in the second line uses the error set bound of $|\Error| \leq O(\gamma_{\ASI}^{-1} d \log^2 n)$ in \Cref{prop:boundASIErrorSet}, the equality there uses our parameter setting $\beta_{\err} = \Theta(d \log n)$ and $\gamma_{\ASI} = 100 \sqrt{d} \log n$, and the final inequality uses $\tau^2 = \Omega(d^{3/2} \log^2 n)$ (with a large enough constant). 

Combining everything, we get $\Pr(L_{t_0}^{\err} \geq \tau^2/6) \leq 1/\poly(n)$ and this completes the proof.  
\end{proof}

\subsubsection{Putting Things Together}
\label{subsec:ellinf_to_ell2_putting_together}

Now we are ready to put things together and prove \Cref{thm:ell_inf-to-ell_2}. 

\begin{proof}[Proof of \Cref{thm:ell_inf-to-ell_2}]
We run the algorithm in \Cref{subsec:ell_inf-to-ell_twoALG}, and we have already shown there that the SDP is always feasible. 
Combining \Cref{fact:quadratic_term}, \Cref{lem:bound_ell_infASI} and \Cref{lem:bound_ellinf_error}, the probability that any of the three (modified) processes $Q_{t}, L_{t}^{\mathsf{ASI}}$ and $L_{t}^{\mathsf{err}}$ violates  \Cref{rule:ell_inf_to_ell_twoSTOP} and \eqref{eq:stopping_condition} is at most $1/\poly(n)$. 
As this stopping condition implies the target discrepancy bound $\|\varphi_t^{\mathcal{P}}\|_2 \leq \tau$,  
by taking a union bound over all time steps and prefixes $\mathcal{P}$, the algorithm never outputs $\mathsf{ABORT}$ with probability at least $1 - 1/\poly(n)$. 
Finally, rounding the coloring of each dead column to $\pm 1$ incurs at most $O(1)$ additive $\ell_2$ discrepancy. This proves the theorem. 
\end{proof}

\section{$\ell_{2}$ to $\ell_{2}$ Prefix Discrepancy}
\label{sec:ell2-to-ell2}

In this section, we prove Theorem \ref{thm:main-ell2-to-ell2}, which is restated below.

\elltwotoelltwo*

We give our algorithm for Theorem \ref{thm:main-ell2-to-ell2} in Section \ref{subsec:ell2-to-ell2ALG} and its analysis in Section \ref{subsec:ell2-to-ell2Analysis}. 
The proof of \Cref{thm:main-ell2-to-ell2} will appear in \Cref{subsec:ell2_to_ell2_putting_together}. 
For our analysis, we also need to use an $\ell_\infty$ bound on the prefix discrepancy $\lVert \varphi_{t}^{\mathcal{P}}\rVert_{\infty}$ for every prefix $\mathcal{P} \in [n]$, which will be given in \Cref{subsec:ell2-toellInfPrefixDiscrepancy}.

\subsection{Algorithm}
\label{subsec:ell2-to-ell2ALG}

Fix $\gamma_{\mathsf{ASI}} = 100\sqrt{d}\,(\log n)^{-1/2}$, and let $\tau := \Theta(\sqrt{d} + d^{1/4}\log^{7/4}n)$ be the target discrepancy bound in Theorem~\ref{thm:main-ell2-to-ell2} (with a large enough constant). Also set $\lambda := \Theta(\log^{3/2}n)$ to be our target $\ell_{\infty}$ prefix discrepancy bound (see \Cref{thm:ell2-to-ellinfMainThm}). Analogous to Algorithm \ref{subsec:ell_inf-to-ell_twoALG}, our algorithm here follows the framework in Sections~\ref{subsec:basic_framework} and \ref{subsecn:updateFramework}, where it uses the $\GlobalIntervalTree$ data structure from Section~\ref{subsec:globalIntervalTree} to maintain a set of at most $\gamma_{\mathsf{ASI}}$ ASI-guarded prefixes $\mathcal{I}_{t} \subseteq W_{t}$.

The main difference from \Cref{subsec:ell_inf-to-ell_twoALG} is that our algorithm also uses a slightly modified version of $\GlobalIntervalTree$, denoted as $\GlobalIntervalTree^\infty$, to choose a subspace $H_t$ to enforce the blocking constraints \eqref{eq:sdp:block-W}. This is only for controlling the $\ell_{\infty}$ prefix discrepancy, and we defer its detail to Section \ref{subsec:ell2-toellInfPrefixDiscrepancy}.

Formally, at each time step $t$, the algorithm does the following.
\begin{enumerate}
    \item If there exists $\mathcal{P} \in [n]$, for which either $\lVert \varphi_{t}^{\mathcal{P}}\rVert_{2}$ exceeds the target $\ell_2$ prefix discrepancy bound $\tau$, or $\lVert \varphi_{t}^{\mathcal{P}}\rVert_{\infty}$ exceeds the $\ell_\infty$ prefix discrepancy bound $\lambda$, it outputs $\mathsf{ABORT}$.
    
    \item Otherwise, it solves the SDP in \eqref{eq:FrameworkSDP}, where $H_{t}$ is the subspace output by $\ModifiedGlobalIntervalTree$ and the matrix $E_{t} \in \mathbb{R}^{d|\mathcal{I}_{t}|\times W_{t}}$ contains a row $A_{i}^{\mathcal{P}}(W_{t})$ for every ASI-guarded prefix $\mathcal{P} \in \mathcal{I}_{t}$ and $i \in [d]$.
\end{enumerate}

Analogous to Algorithm \ref{subsec:ell_inf-to-ell_twoALG}, the row dimension of $E_{t}$ is at most $0.1\gamma_{\mathsf{ASI}}|W_{t}|$ (as we use the same parameters in $\GlobalIntervalTree$). By Proposition \ref{prop:BoundErrorSetInfty}, the dimension of $H_{t}$ is at most $0.1|W_{t}|$. This satisfies the conditions of Fact \ref{fact:sdpFeasibility} and hence SDP \eqref{fact:sdpFeasibility} is always feasible.   In the remainder, we show that the algorithm does not $\mathsf{ABORT}$ with high probability.

\subsection{Analysis}
\label{subsec:ell2-to-ell2Analysis}

Fix any prefix $\mathcal{P} \in [n]$. We will show that, with probability $1 - 1/\mathsf{poly}(n)$, the $\ell_{2}$ prefix discrepancy is bounded as $\lVert \varphi_{t}^{\mathcal{P}}\rVert_{2} \leq \tau$ and the $\ell_\infty$ prefix discrepancy is bounded as $\lVert \varphi_{t}^{\mathcal{P}}\rVert_{\infty} \leq \lambda$ at any time step $t$. Since the algorithm runs for $\mathsf{poly}(n)$ steps and there are $n$ prefixes, a union bound across all time steps and prefixes completes the proof. In the following, we drop the superscript $\mathcal{P}$ whenever it is clear from the context.

\smallskip
\noindent \textbf{Road Map for the Analysis.}
Following the meta-analysis in Section \ref{subsecn:proofFramework} and the analysis in Section \ref{subsec:ell_inf-to-ell_twoAnalysis}, we decompose the
change of $\ell_{2}$ prefix discrepancy $d \lVert \varphi_{t}\rVert_{2}^{2}$ as,
\[
d\lVert \varphi_{t} \rVert_{2}^{2} = dQ_{t} + 2(dL_{t}^{\mathsf{ASI}} + dL_{t}^{\mathsf{err}}).
\]

Our goal will be to show for any time $t \in [0,n]$, with probability $1 - 1/\mathsf{poly}(n)$,
\begin{align} \label{eq:ell2stopping_condition}
Q_t \leq \tau^2/3 \ , \quad L_{t}^{\ASI} \leq \tau^2/6 \ , \quad L_{t}^{\err} \leq \tau^2/6   \ \text{ and } \ \ \lVert\varphi_{t}\rVert_{\infty} \leq \lambda ,
\end{align}
where the extra condition $\lVert\varphi_{t}\rVert_{\infty} \leq \lambda$ is needed for controlling the ASI-guard part $L_t^{\ASI}$. 

Once the above is shown, then by a union bound, one has $\lVert \varphi_{t} \rVert_{2}^{2} \leq \tau^{2}$ and $\lVert \varphi_{t} \rVert_{\infty} \leq \lambda$ for all time steps $t$, which implies that the algorithm doesn't $\mathsf{ABORT}$ with high probability. Analogous to Section \ref{subsec:ell2-to-ell2Analysis}, we use the following stopping condition for our analysis.

\begin{abort}
We freeze $x_{t}$ onward if any of the four conditions in \eqref{eq:ell2stopping_condition}, i.e. $Q_{t} \leq \tau^{2}/3$, $L_{t}^{\mathsf{ASI}} \leq \tau^{2}/6$, $L_{t}^{\mathsf{err}} \leq \tau^{2}/6$, and $\lVert \varphi_{t}\rVert_{\infty} \leq \lambda$ is violated.
\label{abort:ell2ABORT}
\end{abort}

The bound for $Q_{t} \leq O(d + \log n) \leq \tau^{2}/3$ with probability $1 - 1/\mathsf{poly}(n)$ was already shown in \cite{BG17} (see Fact \ref{fact:quadratic_term}). 
In Section \ref{subsubsec:ell2toell2ASI} and Section \ref{subsubsec:ell2toell2Err}, we bound $L_{t}^{\mathsf{ASI}}$ and $L_{t}^{\mathsf{err}}$ in Lemmas \ref{lem:ell2LinearASI} and Lemma \ref{lem:bound_ell2error} respectively.
Finally, in \Cref{subsec:ell2-toellInfPrefixDiscrepancy}, we will describe the modified data structure $\GlobalIntervalTree^\infty$ that allows us to obtain the following bound on $\lVert \varphi_{t} \rVert_{\infty}$.\footnote{We remark that in the current setting, there are algorithms (e.g., \cite{ALS21}) that achieve the better $\ell_{\infty}$ prefix discrepancy bound of $O(\log n)$. However, these algorithms and analysis are quite different from ours, and it is unclear how to achieve the same improvement for the framework that we use here.}

\begin{restatable}[$\ell_\infty$ prefix discrepancy bound]{lemma}{ellinfboundMainThm}
    \label{thm:ell2-to-ellinfMainThm}
Consider the algorithm in Section \ref{subsec:ell2-to-ell2ALG} and the data structure in \Cref{subsec:ell2-toellInfPrefixDiscrepancy}. For any time $t$, and every prefix $\mathcal{P} \in [n]$, with probability $1 - 1/\mathsf{poly}(n)$,
\[
\lVert \varphi_{t}^{\mathcal{P}}\rVert_{\infty} \leq O(\log^{3/2}n).
\]
\end{restatable}

\subsubsection{Bounding the ASI-Guard Part $L_{t}^{\mathsf{ASI}}$}
\label{subsubsec:ell2toell2ASI}

We first show a bound on the $\mathsf{ASI}$-guard term $L_{t}^{\mathsf{ASI}}$ in the following lemma.
\begin{lemma}[Bounding the ASI-guard part, $\ell_2$ to $\ell_2$]    \label{lem:ell2LinearASI}
    Consider the algorithm in Section \ref{subsec:ell2-to-ell2ALG}. For any time $t$, with probability $1 - 1/\mathsf{poly}(n)$,
    \[
    L_{t}^{\mathsf{ASI}} \leq \tau^{2}/6.
    \]
\end{lemma}
Our proof follows the same strategy as Lemmas \ref{lem:bound_ell_infASI} and \ref{lem:bound_ellinf_error}, where we apply Freedman-type concentration for super-martingales. 
\begin{proof}
For a time step $t_{0} \geq 0$, we bound the probability that $L_{t_{0}}^{\mathsf{ASI}}$ exceeds $\tau^{2}/6$.

We let $\mathsf{Corr} = \mathsf{Corr}^{\mathcal{P}} := W_{t^{\star}} \cap \mathcal{P}$ be the set of columns in $W_{t^{\star}}$ with index at most $\mathcal{P}$, where $t^{\star}$ is the first time the window $W_{t^{\star}}$ contains $\mathcal{P}$. Recall from the proof of Lemma \ref{lem:bound_ell_infASI} that $\varphi_{t}$ incurs a non-zero discrepancy only on columns in $\mathsf{Corr}$.
We define the corresponding regularized ASI-guard term as,
\[
    Y_{t} := L_{t}^{\mathsf{ASI}} - \beta_{\mathsf{ASI}}\sum_{j \in \mathsf{Corr}}x_{t}(j)^{2},
\]
where we set $\beta_{\mathsf{ASI}} = O(1)$ for a sufficiently large constant. Below, we show that $dY_{t}$ satisfies the conditions of Freedman-type concentration (Fact \ref{fact:freedman}) for a suitably chosen $\delta$.
Like before, we map time steps $t \in [0,n]$ to $\{0,1,\ldots n/dt\}$.

For time steps $t$, where $x_{t}$ is frozen due to \eqref{eq:ell2stopping_condition} being violated prior to $t$, then $dL_{t}^{\mathsf{ASI}}$ and $dY_{t}$ equal $0$ and trivially satisfy \eqref{eq:freedman_drift_vs_variation} in Fact
\ref{fact:freedman} for any $\delta$. In the other case, we show that condition \ref{eq:freedman_drift_vs_variation} in Fact \ref{fact:freedman} holds for an appropriate $\delta > 0$, by computing the first and the second moment of $dY_{t}$. The increment of $Y_{t}$ is,
\[
    dY_{t} = \langle d\varphi_{t}^{\mathsf{ASI}}, \varphi_{t}\rangle - 2\beta_{\mathsf{ASI}}\sum_{j \in \mathsf{Corr}}u_{t}(j)x_{t}(j)\sqrt{dt} - \beta_{\mathsf{ASI}}\, \sum_{j \in \mathsf{Corr}}u_{t}(j)^{2} dt.
\]
Conditioning on the filtration $\mathcal{F}_{t}$ (or equivalently the outcome of randomness up to time $t$), the first moment of $dY_{t}$ is
\begin{align}
\mathbb{E}_{t}[dY_{t}] = -\beta_{\mathsf{ASI}}\sum_{j \in \mathsf{Corr}}\mathbb{E}_{t}[u_{t}(j)^{2}].
\label{eq:ell2LinearASIFirstMoment}
\end{align}
The second moment of $dY_{t}$ is bounded as,
\begin{align*}
\mathbb{E}_{t}[dY_{t}^{2}] &\lesssim \mathbb{E}_{t} \langle d\varphi_{t}^{\mathsf{ASI}}, \varphi_{t}\rangle^{2} + \beta_{\mathsf{ASI}}^{2}\cdot \mathbb{E}_{t} \bigl(\sum_{j \in \mathsf{Corr}}u_{t}(j)x_{t}(j)\bigr)^{2} dt \\
&\lesssim \mathbb{E}_{t}\langle d\varphi_{t}^{\mathsf{ASI}}, \varphi_{t}\rangle^{2} + \beta_{\mathsf{ASI}}^{2}\sum_{j \in \mathsf{Corr}}\mathbb{E}_{t}[u_{t}(j)^{2}]dt,
\end{align*}
where we drop lower order terms with scale $O((dt)^{3/2})$, use $(a + b)^{2} \leq 2a^{2} + 2b^{2}$ in the first inequality, and use $O(1)$-spectral independence of $u_{t}$ and that $|x_{t}(j)| \leq 1$ in the second inequality. We bound the first term using affine spectral independence.
\begin{align*}
\mathbb{E}_{t}\langle d\varphi_{t}^{\mathsf{ASI}},\varphi_{t}\rangle^{2} &\leq \gamma_{\mathsf{ASI}}\sum_{i \in [d]}\varphi_{t}(i)^{2} \cdot \mathbb{E}_{t} \bigl[d\varphi_{t}^{\mathsf{ASI}}(i)^{2}\bigr] \\
&= \gamma_{\mathsf{ASI}}\cdot \sum_{i \in [d]}\varphi_{t}(i)^{2} \cdot \mathbb{E}_{t}\bigl(\sum_{j \in \mathsf{Corr}}A_{i}(j)u_{t}(j)\bigr)^{2} \lesssim \gamma_{\mathsf{ASI}}\lVert \varphi_{t} \rVert_{\infty}^{2}\sum_{j \in \mathsf{Corr}}\mathbb{E}_{t}[u_{t}(j)^{2}]dt,
\end{align*}
where the first step uses $\gamma_{\mathsf{ASI}}$-affine spectral independence of $u_{t}$, and the final step uses that $u_t$ is $O(1)$-spectral independence, $\sum_{i \in [d]}A_{i}(j)^{2} \leq 1$ for all $j \in [n]$ and the $\ell_{1}$-$\ell_{\infty}$ Hölder's inequality.

Combining, we get the following bound on the second moment of $dY_{t}$,
\begin{align}
\mathbb{E}_{t}[dY_{t}^{2}] \lesssim (\gamma_{\mathsf{ASI}}\lVert \varphi_{t} \rVert_{\infty}^{2} + \beta_{\mathsf{ASI}}^{2})\sum_{j \in \mathsf{Corr}}\mathbb{E}_{t}[u_{t}(j)^{2}]dt \leq (\gamma_{\mathsf{ASI}}\lambda^{2} + \beta_{\mathsf{ASI}}^{2})\sum_{j \in \mathsf{Corr}}\mathbb{E}_{t}[u_{t}(j)^{2}]dt,
    \label{eq:ell2LinearASISecondMomentBound}
\end{align}
where we bound $\lVert \varphi_{t} \rVert_{\infty}^{2} \leq \lambda^{2}$ by \Cref{thm:ell2-to-ellinfMainThm}.
Combining \eqref{eq:ell2LinearASIFirstMoment} and \eqref{eq:ell2LinearASISecondMomentBound}, we get condition \eqref{eq:freedman_drift_vs_variation} with,
\[
\mathbb{E}_{t}[dY_{t}] \lesssim -\bigl(\frac{\beta_{\mathsf{ASI}}}{\lambda^{2} \gamma_{\mathsf{ASI}} + \beta_{\mathsf{ASI}}^{2}}\bigr)\mathbb{E}_{t}[dY_{t}^{2}].
\]
Then applying Freedman-type inequality (Fact \ref{fact:freedman}), with deviation $\xi = \Theta(\gamma_{\mathsf{ASI}}\log^{4}n)$, we obtain
\begin{align}
    \Pr(Y_{t} - Y_{t^{\star}} \geq \xi) \leq \exp(\frac{-\xi \, \beta_{\mathsf{ASI}}}{\beta_{\mathsf{ASI}}^{2} + \gamma_{\mathsf{ASI}}\lambda^{2}}).
    \label{eq:ell2FreedmanBoundLinearASI}
\end{align}
To bound the exponent in \eqref{eq:ell2FreedmanBoundLinearASI}, note that since $\xi = \Theta(\gamma_{\mathsf{ASI}}\log^{4}n) = \Theta(\sqrt{d}\log^{3.5}n)$ and $\beta_{\mathsf{ASI}} = O(1)$ (for sufficiently large constants), the first term is bounded as $(\xi \beta_{\mathsf{ASI}})/\beta_{\mathsf{ASI}}^{2} = \Omega(\log n)$. The second term can be bounded as,
\[
\frac{\xi\beta_{\mathsf{ASI}}}{\gamma_{\mathsf{ASI}}\lambda^{2}} = \frac{\Theta(\gamma_{\mathsf{ASI}}\log^{4}n)}{\gamma_{\mathsf{ASI}} \cdot O(\log^{3}n)} = \Theta(\log n).
\]
Thus, the tail probability in \eqref{eq:ell2FreedmanBoundLinearASI} is at most $1/\mathsf{poly}(n)$ (provided that the constants in $O(\cdot)$ for $\beta_{\mathsf{ASI}}$ and $\xi$ are sufficiently large). Finally, as in the proof of \Cref{lem:bound_ell_infASI}, the tail bound for \(Y_{t_0} - Y_{t^{\star}} \geq \xi\) implies a bound on the probability of our target event \(L_{t_0}^{\mathsf{ASI}} \geq \tau^{2}/6\). Note that as before, 

\[
Y_{t_{0}} - Y_{t^{\star}} \geq L_{t_{0}}^{\mathsf{ASI}} - \beta_{\mathsf{ASI}}(10d).
\]
As $\tau^{2} = \Theta(d + \sqrt{d}\log^{3.5}n)$ (with a much larger constant), $\tau^{2}/6 \geq \beta_{\mathsf{ASI}}(10d) + \xi$. Consequently, the event $L_{t_{0}}^{\mathsf{ASI}} \geq \tau^{2}/6$ implies that $Y_{t_{0}} - Y_{t^{\star}} \geq \xi$ and thus,
\[
\Pr(L_{t_{0}}^{\mathsf{ASI}} \geq \tau^{2}/6) \leq \Pr(Y_{t_{0}} - Y_{t^{\star}} \geq \xi) \leq 1/\mathsf{poly}(n).
\]
Taking a union bound over all time steps $t_{0}$ completes the proof of the lemma.

\end{proof}

\subsubsection{Bounding the Error Part $L_{t}^{\mathsf{err}}$}
\label{subsubsec:ell2toell2Err}

Next, we bound the error term $L_{t}^{\mathsf{err}}$ in the following lemma.

\begin{lemma}[Bounding the error part, $\ell_2$ to $\ell_2$] \label{lem:bound_ell2error}
    Consider the algorithm in Section \ref{subsec:ell2-to-ell2ALG}. With probability $1 - 1/\mathsf{poly}(n)$, we have
    \[
    L_{t}^{\mathsf{err}} \leq \tau^{2}/6.
    \]
\end{lemma}

\begin{proof}
    For $t_{0} \geq 0$, we bound the probability that $L_{t_{0}}^{\mathsf{err}}$ exceeds $\tau^{2}/6$.

    Analogous to the proof of Lemma \ref{lem:bound_ellinf_error}, we define the sets $\mathsf{Error}^{\mathcal{P}}$ (which we denote by $\mathsf{Error}$ for shorthand). The regularized process for $L_{t}^{\mathsf{err}}$:
    \begin{align}
        Z_{t} = L_{t}^{\mathsf{err}} - \beta_{\mathsf{err}}\sum_{j \in \mathsf{Error}_{t}^{\mathcal{P}}}x_{t}(j)^{2},
    \end{align}
    where $\beta_{\mathsf{err}} = O(\log n)$.
    Analogous to the proofs of Lemma \ref{lem:bound_ell_infASI}, Lemma \ref{lem:bound_ellinf_error} and Lemma \ref{lem:ell2LinearASI}, we analyze the increments at a time $t$ before the process has frozen.
    The increment of $dZ_{t}$ is given by,
    \begin{align}
        dZ_{t} = \langle d\varphi_{t}^{\mathsf{err}}, \varphi_{t}\rangle - 2\beta_{\mathsf{err}}\sum_{j \in \mathsf{Error}_{t}}u_{t}(j)x_{t}(j)\sqrt{dt} - \beta_{\mathsf{err}}\sum_{j \in \mathsf{Error}_{t}}u_{t}(j)^{2}.
    \end{align}

    Since $d\varphi_{t}^{\mathsf{err}}$ and $u_{t}$ are mean-zero, the first moment of $dZ_{t}$ is:
    \begin{align}
        \mathbb{E}_{t}[dZ_{t}] = -\beta_{\mathsf{err}}\sum_{j \in \mathsf{Error}_{t}}\mathbb{E}_{t}[u_{t}(j)^{2}].
        \label{eq:ell2LinearErrFirstMoment}
    \end{align}

    The second moment can be bounded as,
    \begin{align*}
    \mathbb{E}_{t}[(dZ_{t})^{2}] &= \mathbb{E}_{t}\bigl( \langle d\varphi_{t}^{\mathsf{err}}, \varphi_{t}\rangle - 2\beta_{\mathsf{err}}\sum_{j \in \mathsf{Error}_{t}}u_{t}(j)x_{t}(j)\sqrt{dt}\bigr)^{2}\\
    &\lesssim \mathbb{E}_{t}[\langle d\varphi_{t}^{\mathsf{err}},\varphi_{t}\rangle^{2}] + \beta_{\mathsf{err}}^{2}\cdot \mathbb{E}_{t}\bigl(\sum_{j\in \mathsf{Error}_{t}} u_{t}(j)x_{t}(j)\bigr)^{2} dt\\
    &\lesssim \mathbb{E}_{t}[\langle d\varphi_{t}^{\mathsf{err}},\varphi_{t}\rangle^{2}] + \beta_{\mathsf{err}}^{2}\sum_{j \in \mathsf{Error}_{t}}\mathbb{E}[u_{t}(j)^{2}],
    \end{align*}
    where we ignore the terms of scale $O(dt^{3/2})$ as they are lower order and the second inequality follows via $O(1)$-spectral independence of $u_{t}$. We bound the first term as follows:
    \begin{align*}
    \mathbb{E}_{t}[\langle d\varphi_{t}^{\mathsf{err}}, \varphi_{t}\rangle^{2}] &= \mathbb{E}_{t}\bigl(\sum_{j \in \mathsf{Error}_{t} \cap \mathcal{V}_{t}}u_{t}(j)((W_{t}^{\mathcal{P}})^{\top}\varphi_{t})(j)\bigr)^{2} \\
    &\lesssim \sum_{j \in \mathsf{Error}_{t}\cap \mathcal{V}_{t}}\mathbb{E}_{t}\bigl[u_{t}(j)^{2}((W_{t}^{\mathcal{P}})^{\top}\varphi_{t})(j)^{2}\bigr] \leq \lVert \varphi_{t}\rVert_{2}^{2} \sum_{j \in \mathsf{Error}_{t} \cap \mathcal{V}_{t}}\mathbb{E}[u_{t}(j)^{2}].
    \end{align*}
    Combining we get,
    \begin{align}
    \mathbb{E}_{t}[dZ_{t}^{2}] \lesssim (\lVert \varphi_{t}\rVert_{2}^{2} + \beta_{\mathsf{err}}^{2})\sum_{j \in \mathsf{Error}_{t} \cap \mathcal{V}_{t}}\mathbb{E}[u_{t}(j)^{2}]\ \lesssim (\tau^{2} + \beta_{\mathsf{err}}^{2})\sum_{j \in \mathsf{Error}_{t} \cap \mathcal{V}_{t}}\mathbb{E}[u_{t}(j)^{2}],
        \label{eq:ell2LinearErrSecondMomentBound}
    \end{align}
where the second inequality follows as the process has not frozen yet.
    This satisfies the Condition \eqref{eq:freedman_drift_vs_variation} in Fact \ref{fact:freedman} as,
    \[
    \mathbb{E}_{t}[dZ_{t}] \lesssim \bigl(\frac{-\beta_{\mathsf{err}}\, \xi}{\beta_{\mathsf{err}}^{2} + \tau^{2}}\bigr)\mathbb{E}_{t}[dZ_{t}^{2}].
    \]

Applying Freedman-type inequality for deviation $\xi = \tau^{2}/10$, we get
\[
\Pr(Z_{t} - Z_{0} \geq \xi) \leq \exp(\frac{-\xi\, \beta_{\mathsf{err}}}{\beta_{\mathsf{err}}^{2} + \tau^{2}}) = 1/\mathsf{poly(n)},
\]

where the last inequality is obtained by using that $\beta_{\mathsf{err}} = O(\log n)$ and $\xi = \tau^{2}/10$, where $\tau^{2} = O(d + \sqrt{d}\log^{7/2}n)$. This implies that $(\xi \beta_{\mathsf{err}})/\beta_{\mathsf{err}}^{2} = \Omega(d^{1/2}\log^{7/2}n/\log n)$ which is  $\Omega(\log n)$ and $(\xi\beta_{\mathsf{err}}/\tau^{2}) = \beta_{\mathsf{err}}/10 = \Omega(\log n)$ provided the constant in $\beta_{\mathsf{err}}$ is large enough.

Finally, similar to the proof of Lemma \ref{lem:bound_ellinf_error}, the event $L_{t_{0}}^{\mathsf{err}} \geq \tau^{2}/6$ implies the event $Z_{t_{0}} - Z_{t^{\star}} \geq \tau^{2}/10$. In particular, using $\tau^{2} = \Theta(d + d^{1/2}\log^{7/2}n)$ for a large enough constant, we have 
\[
Z_{t_{0}} - Z_{t^{\star}} \geq L_{t_{0}}^{\mathsf{err}} - \beta_{\mathsf{err}}\cdot O(\gamma_{\mathsf{ASI}}^{-1}d\log^{2} n) = L_{t_{0}}^{\mathsf{err}} - O(1) \cdot \frac{d \log^{3}n}{\sqrt{d}(\log n)^{-1/2}} \geq L_{t_{0}}^{\mathsf{err}} - \tau^{2}/10 .
\]

Combining everything, we get $\Pr(L_{t_{0}}^{\mathsf{err}} \geq \tau^{2}/6) \leq 1/\mathsf{poly}(n)$, and this completes the proof. 
\end{proof}

\subsubsection{Putting Things Together}
\label{subsec:ell2_to_ell2_putting_together}

\begin{proof}[Proof of \Cref{thm:main-ell2-to-ell2}]
Consider the algorithm in \Cref{subsec:ell2-to-ell2ALG}. There, we have already shown that the SDP at each step of the algorithm is feasible.
Combining \Cref{fact:quadratic_term}, \Cref{thm:ell2-to-ellinfMainThm}, \Cref{lem:ell2LinearASI} and \Cref{lem:bound_ell2error}, the probability that any of the four conditions in \eqref{eq:ell2stopping_condition} is violated is at most $1/\poly(n)$. 
by taking a union bound over all time steps and prefixes $\mathcal{P}$, the algorithm never outputs $\mathsf{ABORT}$ with probability at least $1 - 1/\poly(n)$. 
Finally, rounding the coloring of each dead column to $\pm 1$ incurs at most $O(1)$ additive $\ell_2$ discrepancy, and this completes the proof. 
\end{proof}

\subsection{Bounding $\ell_{\infty}$ Prefix Discrepancy}
\label{subsec:ell2-toellInfPrefixDiscrepancy}

In this section, we describe the data structure $\GlobalIntervalTree^\infty$ and prove \Cref{thm:ell2-to-ellinfMainThm}, restated below. 

\ellinfboundMainThm*

To bound the $\ell_\infty$ prefix discrepancy, we will enforce a set of blocking constraints (corresponding to $H_t$ in \eqref{eq:sdp:block-rows}). 
We describe these constraints and how it controls $\|\varphi_t^{\mathcal{P}}\|_\infty$ in \Cref{subsubsec:DiscBeyondPartialColoring}. 
Then in \Cref{subsubsec:modifiedGlobalIntervalTree}, we present the modified data structure $\ModifiedGlobalIntervalTree$ for maintaining these blocking constraints and analyze its properties. Finally, we prove  \Cref{thm:ell2-to-ellinfMainThm} in \Cref{subsubsec:proof-ell2-to-ellinfMainThm}.

\subsubsection{Our Strategy}\label{subsubsec:DiscBeyondPartialColoring}

We follow the approach in \cite{BG17} of bounding the prefix discrepancy $|\varphi_t^{\mathcal{P}}(i)|$ for each row $i \in [d]$. 
Recall that $\varphi_t^{\mathcal{P}}(i)$ only incurs non-zero discrepancy when $\mathcal{P} \in W_t$ (due to constraint \eqref{eq:sdp:block-W}). 
To control its discrepancy while $\mathcal{P} \in W_t$, we impose a set of blocking constraints (i.e., $H_t \neq \emptyset$  in \eqref{eq:sdp:block-rows}).
We fix an arbitrary row $i \in [d]$ and describe the blocking constraints for row $i$ below.

At each time step $t$, we maintain a subset $\mathcal{J}_{i,t} \subseteq W_t$ and include the prefix $A_i^{\mathcal{P}'}$ of row $i$ into the subspace $H_t$ for each $\mathcal{P}' \in \mathcal{J}_{i,t}$, i.e., we ensure that $d \varphi_t^{\mathcal{P'}}(i) = 0$ for any $\mathcal{P}' \in \mathcal{J}_{i,t}$. 
We require $\sum_{i \in [d]} |\mathcal{J}_{i,t}| \leq 0.1 |W_t|$ to ensure that $\dim(H_t) \leq 0.1|W_t|$, as needed in \Cref{fact:sdpFeasibility}. 
The data structure $\ModifiedGlobalIntervalTree$ that maintains the sets $\{\mathcal{J}_{i,t}\}_{i \in [d]}$ will be given in \Cref{subsubsec:modifiedGlobalIntervalTree}. Here, we show to bound $|\varphi_t^{\mathcal{P}}(i)|$, and we drop the superscript $\mathcal{P}$ below when the context is clear. 

\smallskip
\noindent
\textbf{Bad Columns for $\mathcal{P}$, and Bounding $|\varphi_t(i)|$.} 
Borrowing terminology from \Cref{subsec:globalIntervalTree}, we call $ \mathcal{J}_{i,t}$ the {\em blocking-guarded prefixes} for row $i$, and call the remaining prefixes {\em unguarded} for row $i$. 
For each unguarded prefix $\mathcal{P} \in [W_t]$ for row $i$, we call $\max\{\mathcal{J}_{i,t} \cap [\mathcal{P}]\}$, i.e., the maximum-indexed blocking-guarded prefix prior to $\mathcal{P}$, its {\em blocking-guard} for row $i$. The alive columns between the blocking-guard and $\mathcal{P}$ will be called {\em bad columns}. Analogous to \Cref{defn:error-set}, we define the following notation of {\em bad set} for prefix $\mathcal{P}$ in row $i$. 

\begin{definition}[Bad set for $\mathcal{P}$ in row $i$] 
    The bad set for prefix $\mathcal{P}$ in row $i$ up to time $t$, denoted as $\Bad_{i,t}^{\mathcal{P}}$, is the union of all bad columns for $\mathcal{P}$ at all time steps $t' \leq t$. The bad set for $\mathcal{P}$ is defined as $\Bad^{\mathcal{P}}_{i} := \Bad^{\mathcal{P}}_{i,n}$. 
\end{definition}
Note that the blocking-guard of prefix $\mathcal{P}$ incurs zero discrepancy change at time $t$, and thus
\[
d \varphi_t(i) = d \varphi_t^{\bad}(i) ,
\]
where $d \varphi_t^{\bad}(i)$ is the discrepancy change due to the bad columns for row $i$ at time $t$ (see \Cref{fig:blockGuard}).

\begin{figure}[h]
    \centering
    \includegraphics[width=0.6\textwidth]{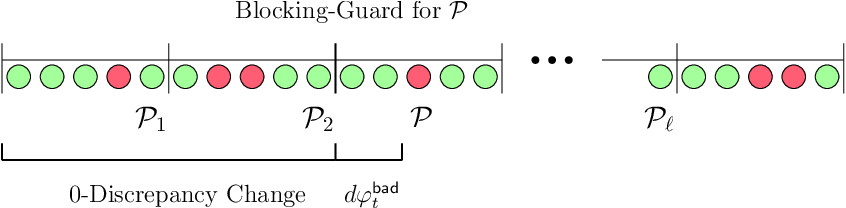}
    \caption{The blocking-guard for prefix $\mathcal{P}$ in row $i$. The discrepancy change $d \varphi_t^{\mathcal{P}}(i) = d \varphi_t^{\bad}(i)$ is entirely due to the set of bad columns.}

    \label{fig:blockGuard}
\end{figure}

It turns out that the contribution of the bad columns depends on their {\em $\ell_2^2$ mass}, where we define the $\ell_2^2$ mass of a set $S \subseteq [n]$ in row $i$ to be $\|A_i(S)\|_2^2= \sum_{j \in S} A_i(j)^2$.  
In particular, \cite[Theorem 1]{BG17} gives the bound $|\varphi_t(i)| \leq O(\sqrt{\|A_i(\Bad_i)\|_2^2 \cdot \log n})$ with probability $1- 1/\poly(n)$.

Therefore, to obtain the result in \Cref{thm:ell2-to-ellinfMainThm}, we need to bound $\|A_i(\Bad_i)\|_2^2 \leq O(\log^2 n)$. We will show how $\ModifiedGlobalIntervalTree$ achieves this condition in \Cref{subsubsec:modifiedGlobalIntervalTree}.

\subsubsection{Modified Global Interval Tree}
\label{subsubsec:modifiedGlobalIntervalTree}
In this section, we describe the data structure $\ModifiedGlobalIntervalTree$ for maintaining the blocking constraints $\{\mathcal{J}_{i,t}\}_{i \in [d]}$. 

The idea is similar to $\GlobalIntervalTree$, but we include a full description here for completeness.

Similar to the interval representation in \Cref{subsec:globalIntervalTree}, the set of blocking-guarded prefixes $\mathcal{J}_{i,t}$ (for row $i$) corresponds naturally to consecutive intervals that partition $[W_{t}]$. We abuse notation and also use $\mathcal{J}_{i,t}$ to denote the set of intervals in this partition. 
$\ModifiedGlobalIntervalTree$ maintains these intervals $\{\mathcal{J}_{i,t}\}_{i \in [d]}$ (that satisfy $\sum_{i \in [d]} |\mathcal{J}_{i,t}| \leq 0.1 |W_t|$) by adding and merging (small) intervals as before.

\smallskip
\noindent
\textbf{Initialization and Adding Intervals.} Before the process starts, for each row $i \in [d]$, we split $[n]$ into a collection $\mathcal{B}_{i,0}$ of base intervals, each with $\ell_{2}$-squared mass $s_0 = 20  \log n$.

Then, $\mathcal{J}_{i,0}$ will be the set of all base intervals that are completely contained in the first $10d$ columns.\footnote{We do not need an extra constraint at the end of the window as in \Cref{remark:WindowEndConstraint}, since the entire time window $W_t$ is blocked at each time step $t$ due to \eqref{eq:sdp:block-rows}.} 

Note that the total $\ell_2^2$ mass of $W_0$ over all $d$ rows is at most $|W_0|$ (as each column has $\ell_2$ norm at most $1$), we have $\sum_{i \in [d]} |\mathcal{J}_{i,0}| \leq 0.1 |W_0|$. 
As before, the data structure adds a base interval in $\mathcal{B}_{i,0}$ to $\mathcal{J}_{i,t}$ whenever it becomes active, i.e., it contains an active column and no dormant column.

\smallskip
\noindent
\textbf{Merging Intervals.}
Analogous to Section \ref{subsec:globalIntervalTree}, the intervals in $\mathcal{J}_{i,t}$ are merged using a global binary tree $\mathsf{Tree}_{i,t}$. Initially, $\mathsf{Tree}_{i,0}$ is the complete binary tree of height $O(\log n)$ whose leaves correspond to all the base intervals $\mathcal{B}_{i,0}$, and at each time $t$, $\mathsf{Tree}_{i,t}$ is a subtree of $\mathsf{Tree}_{i,0}$. As before, the active leaves of $\mathsf{Tree}_{i,t}$ correspond to intervals in $\mathcal{J}_{i,t}$. 

As before, sibling leaves of $\mathsf{Tree}_{i,t}$ are merged carefully when they become {\em small}, defined as follows. The {\em $\ell_2^2$ size} of an interval $I$ in row $i$ at time $t$ is defined to be the $\ell_2^2$ mass (in row $i$) of the alive columns contained in $I$ at time $t$.
We call an interval {\em small} if its $\ell_2^2$ size is $\leq s_0/2 = 10 \log n$ at time $t$. 

Analogous to \Cref{mergerule:global_binary_tree}, $\mathsf{Tree}_{i,t}$ is updated via the following merging rule.

\begin{mergerule}[Merging via $\ModifiedGlobalIntervalTree$] \label{rule:mergeRuleModifiedGlobalTree} At any time $t$,
\begin{enumerate}
    \item For any small left leaf $v \in \mathsf{Tree}_{t}$, delete it and merge its interval with the interval at the next leaf; In particular, if the right sibling of $v$ is also a leaf, then also delete the right siblind and use their parent to represent the merged interval.
    
    \item For any small right leaf $v \in \mathsf{Tree}_{i,t}$, if their left sibling of $v$ has been deleted, then delete $v$ and use its parent to represent its interval.
\end{enumerate}
\end{mergerule}

We end this subsubsection by proving the correctness of $\ModifiedGlobalIntervalTree$. In \Cref{prop:BlockingConstraintBound}, we bound the number of intervals in $\{\mathcal{J}_{i,t}\}_{i \in [d]}$ and in \Cref{prop:BoundErrorSetInfty}, we show that $\mathsf{Bad}_{i}^{\mathcal{P}}$ has small $\ell_2^2$ mass.

\begin{proposition}[Number of blocking constraints]     \label{prop:BlockingConstraintBound}At any time $t$, consider $\mathsf{Tree}_{i,t}$ after \Cref{rule:mergeRuleModifiedGlobalTree} is completed. Then the average $\ell_2^2$ size of active leaves in $\mathsf{Tree}_{i,t}$ is at least $s_0/(2 \log n) = 10$. Consequently, the total number of blocking constraints is    $\big|\sum_{i \in [d]}\mathcal{J}_{i,t} \big| \leq 0.1|W_{t}|$.
\end{proposition}
\begin{proof}
The proof is essentially the same as the proof of Proposition \ref{prop:boundASIConstraints}. 
As before, there is no small left leaf. By the same charging argument there (that charges $<\log n$ small right leaves to any active non-small left leaf), the average $\ell_2^2$ size of the active leaves (or intervals in $\mathcal{J}_{i,t}$) is at least $s_0/(2 \log n) = 10$. The second statement in the proposition then follows immediately, by noting that the total $\ell_2^2$ size of the time window $W_t$ is at most $|W_t|$. 

\end{proof}

Next we show how to bound the $\ell_2^2$ mass of the bad set for $\mathcal{P}$ in row $i$. 

\begin{proposition}[Bounding the bad set]     \label{prop:BoundErrorSetInfty}
The $\ell_{2}^{2}$ mass of $\mathsf{Bad}^{\mathcal{P}}_{i}$ is at most $(s/2) \log n = O(\log^{2}n)$.
\end{proposition}

\begin{proof}
This proof is essentially the same as the proof of \Cref{prop:boundASIErrorSet}, by decomposing the prefix $\mathcal{P}$ into a union of $<\log n$ subtrees (for row $i$), and show that each subtree contributes at most $s/2$ to the $\ell_2^2$ mass of $\mathsf{Bad}^{\mathcal{P}}_{i}$. We omit the details. 

\end{proof}

\subsubsection{Proof of \Cref{thm:ell2-to-ellinfMainThm}}
\label{subsubsec:proof-ell2-to-ellinfMainThm}
Now we are ready to complete the proof of \Cref{thm:ell2-to-ellinfMainThm}. 

\begin{proof}[Proof of \Cref{thm:ell2-to-ellinfMainThm}]
Fix a prefix $\mathcal{P}$ and a row $i \in [d]$. 
\Cref{prop:BlockingConstraintBound} ensures that $\sum_{i \in [d]} |\mathcal{J}_{i,t}| \leq 0.1 |W_t|$ at each time $t$, so that the condition of \Cref{fact:sdpFeasibility} is met and the algorithm is feasible. 
By \Cref{prop:BoundErrorSetInfty},  the $\ell_{2}^{2}$ mass of $\Bad_i^{\mathcal{P}}$ is at most $O(\log^2 n)$. 
Then by \cite[Theorem 1]{BG17}, with probability $1 - 1/\poly(n)$, all time $t$ one has  
\[
|\varphi_t^{\mathcal{P}}(i)| \leq O \Big(\sqrt{\big\|A_i(\Bad_i^{\mathcal{P}}) \big\|_2^2 \cdot \log n} \Big) \leq O(\log^{3/2} n). 
\]
Taking a union bound over all prefixes $\mathcal{P} \in [n]$ and rows $i \in [d]$ completes the proof. 

\end{proof}

\section{Concluding Remarks}
\label{sec:conclusion_open_problems}

Our constructive bounds for $\ell_2$ prefix discrepancy and Steinitz problems match the conjectured bounds when $d \geq \polylog(n)$, where the $\polylog$ factors are larger than those in Banaszczyk's non-constructive results (where he only requires $d \geq \log n$). 
Below, we briefly explain the several reasons for this loss in our current analysis and mention some related open problems. 

\begin{enumerate}
    \item[(a)] First, in the $\GlobalIntervalTree$ data structure, we lower bound the average interval size by $\Omega(s/\log n)$ (even though the base intervals have size $s$), which results in an extra $\log n$ factor in the number of intervals we need to maintain; we also upper bound the size of the error set by $O(s \log n)$. Both these $\log n$ factors are due to the global binary tree having depth $\Theta(\log n)$. Nonetheless, as the sliding window has size $|W_t| = O(d)$, it may be possible to reduce these factors to $\log d$ if we could maintain a binary tree structure within $W_t$. 

    \item[(b)] Second, as we are controlling $\gamma_{\ASI} = \widetilde{\Theta}(\sqrt{d})$ many ASI-guarded prefixes simultaneously, this leads to an extra $\widetilde{\Theta}(d^{1/4})$ factor in the additive term (e.g., see \eqref{eq:ell_infty-to-ell_2-overview}).
    Getting rid of this factor would require attaining $O(1)$-affine spectral independence (or equivalently, satisfying \eqref{eq:sdp:ASI} for only $O(1)$ prefixes), and it is unclear how to achieve this. 

    \item[(c)] Third, we used the $\ell_\infty$ prefix discrepancy bound of $O(\log^{3/2} n)$ in \Cref{thm:ell2-to-ellinfMainThm}. 
    The extra $\sqrt{\log n}$ loss (from the current best constructive bound of $O(\log n)$ in \cite{ALS21}) is due to the $\Theta(\log n)$ depth of $\ModifiedGlobalIntervalTree$, and it may be possible to replace it by a $\log d$ factor as mentioned in (a). Alternatively, it would be interesting to see if the (rather) different ideas in \cite{ALS21} could be combined with our algorithmic framework to achieve a $O(\log n)$ bound directly.

\end{enumerate}

We now mention some open problems related to prefix discrepancy below. 

\medskip
\noindent \textbf{Matching Banaszczyk's Bound Algorithmically.} It is an intriguing open question to match Banaszczyk's non-constructive bound for $\ell_2$ to $\ell_2$ prefix discrepancy. 

\begin{question}[Matching Banaszczyk's bound for $\ell_2$ to $\ell_2$ prefix discrepancy]
    Given vectors $v_{1},\ldots, v_{n} \in \mathbb{R}^{d}$ with $\lVert v_{i}\rVert_{2} \leq 1$ for each $i \in [n]$, does there exist an efficient algorithm that finds $x \in \{\pm 1\}^{n}$ such that $\lVert \sum_{i = 1}^{t} x_{i}v_{i}\rVert_{2} \leq O(\sqrt{d} + \sqrt{\log n})$ for all prefix $t \in [n]$?
\end{question}

Similarly, for the $\ell_2$ to $\ell_\infty$ prefix discrepancy (i.e., the prefix version of Koml\'os problem), it is widely open how to attain Banaszczyk's $O(\sqrt{\log n})$ bound algorithmically. As mentioned earlier, the current best algorithmic bound is $O(\log n)$ \cite{ALS21}.

\begin{question}[Matching Banaszczyk's bound for prefix Koml\'os]
Given $v_{1},v_{2},\ldots v_{n} \in \mathbb{R}^{d}$ with $\lVert v_{i}\rVert_{2} \leq 1$ for each $i \in [n]$, does there exist an efficient algorithm that can find $x \in \{\pm 1\}^{n}$ such that $\lVert \sum_{i=1}^{t} x_{i}v_{i}\rVert_{\infty} \leq O(\sqrt{\log n})$ for all prefix $t \in [n]$?
\end{question}

More generally, any algorithm that can find a random coloring $x \in \{\pm 1\}^n$ such that every prefix discrepancy vector is $O(1)$-subgaussian can attain both bounds above (in fact, this prefix $O(1)$-subgaussian property is equivalent to Banaszczyk's result in \cite{Ban12}). 
Currently, the best algorithm can only achieve $O(\log n)$-subgaussianity for every prefix discrepancy \cite{ALS21}.

\medskip
\noindent \textbf{Beating Banaszczyk's Bound for Prefix Discrepancy.} For some time, one of the central goals of algorithmic discrepancy has been to match Banaszczyk's non-constructive bound algorithmically. But recently, \cite{BJ25b,BJ26} were able to go beyond and substantially improve upon Banaszczyk's result for the Beck-Fiala and Koml\'os problems. 
It would be very interesting (and significant) to improve upon Banaszczyk's non-constructive bounds for prefix discrepancy problems.

\begin{question}[Beating Banaszczyk's bound for $\ell_2$ to $\ell_2$ prefix discrepancy]
    Given $v_{1},\ldots, v_{n} \in \mathbb{R}^{d}$ with $\lVert v_{i}\rVert_{2} \leq 1$ for each $i \in [n]$, does there exist $x \in \{\pm 1\}^{n}$ with $\lVert \sum_{i = 1}^{t} x_{i}v_{i}\rVert_{2} \leq O(\sqrt{d}) + o(\sqrt{\log n})$?
\end{question}

Similarly, it is an open question to beat Banaszczyk's $O(\sqrt{d \log n})$ bound for the $\ell_\infty$ to $\ell_\infty$ prefix discrepancy problem. Interestingly, unlike the $\ell_2$ to $\ell_2$ setting, an algorithm matching Banaszczyk's bound in this setting is known \cite{BG17}. 

\begin{question}[Beating Banaszczyk's bound for $\ell_\infty$ to $\ell_\infty$ prefix discrepancy]
    Given $v_{1},\ldots, v_{n} \in \mathbb{R}^{d}$ with $\lVert v_{i}\rVert_{\infty} \leq 1$ for each $i \in [n]$, does there exist $x \in \{\pm 1\}^{n}$ with $\lVert \sum_{i = 1}^{t} x_{i}v_{i}\rVert_{\infty} \leq o(\sqrt{d \log n})$?
\end{question}

Another closely related question is Tusn\'ady's problem, which asks for the combinatorial discrepancy of all axis-parallel rectangles for an arbitrary set of $n$ points in $[0,1]^d$. 
For this problem, the current best lower bound is $\Omega(\log^{d-1} n)$ \cite{MN15}, while the current best non-constructive bound (which crucially relies on Banaszczyk's result) is $O(\log^{d-1/2} n)$ \cite{Nik17}, and the current best constructive bound is $O(\log^d n)$ \cite{BG17}. 
It would be quite interesting to improve any of these results.

\newpage

\appendix

\section{SDP Feasibility}
\label{appendix:SDPFeasibility}

In this subsection, we give the statement of the SDP feasibility theorem in \cite[Theorem A.4]{BJ26} and explain how \Cref{fact:sdpFeasibility} follows from it. 

\begin{theorem}[Theorem A.4 of \cite{BJ26}]\label{thm:SDPFeasibility}
Let $W \subset \mathbb{R}^{h}$ be a subspace with dimension $\mathrm{dim}(W) = \delta h,$ and $E_{s} \in \mathbb{R}^{r_{s}h\times h}$ for all $s \in [q]$ be a set of matrices with $r_{s} \geq 1$. Then, for any $0 \leq \kappa, \eta, \eta_{s} < 1$, where $s \in [q]$, such that $\eta + \kappa + \sum_{s = 1}^{q}\eta_{s} \leq 1 - \delta$, there is an $h\times h$ PSD matrix $U$ satisfying:
\begin{enumerate}
    \item $U_{j,j} \leq 1$ for all $j \in [h]$,
    \item $\mathrm{Tr}(U) \geq \kappa h$,
    \item $\langle ww^{\top}, U\rangle$ for all $w \in W$,
    \item $U \preceq \frac{1}{\eta}\mathsf{diag}(U)$, and
    \item $E_{s}UE_{s}^{\top} \preceq \frac{r_{s}}{\eta_{s}}\mathsf{diag}(E_{s}UE_{s}^{\top})$ for all $s \in [q]$.
\end{enumerate}
Furthermore, such a PSD matrix $U$ can be computed by solving a semidefinite program (SDP).
\end{theorem}

To obtain \Cref{fact:sdpFeasibility} from \Cref{thm:SDPFeasibility}, we set the parameters in \Cref{thm:SDPFeasibility} for our SDP \eqref{eq:FrameworkSDP} as follows. 
We set $\kappa = 0.1$ corresponding to the trace condition in \eqref{eq:sdp:trace}.

\textbf{Blocking Constraints.} The subspace $W$ in \Cref{thm:SDPFeasibility} will contain the subspace $H_t$ for the blocking constraints in \eqref{eq:sdp:block-rows}, which has $\dim(H_t) \leq 0.1|W_{t}|$, as well as the $d$ rows of the sliding window $W_t$ provided $|W_{t}| = 10d$ (as mentioned in \Cref{footnote:drop_sliding_window_constrarint}, if $|W_t| < 10d$, the constraint \eqref{eq:sdp:block-W} will be dropped). 

This corresponds to setting $\delta$ to be at most $0.2$ in \Cref{thm:SDPFeasibility}.

\textbf{SI Constraint.} The SI constraint \eqref{eq:sdp:SI} corresponds to setting $\eta = 0.1$  in \Cref{thm:SDPFeasibility}.

\textbf{ASI Constraint.} The constraint \eqref{eq:sdp:ASI} corresponds to setting $q = 1$ and $\eta_1 = 0.1$ in \Cref{thm:SDPFeasibility}, as the ratio between the row and column dimensions of $E_t$ is $r_1 = 0.1 \gamma_{\ASI}$. 

The parameter setting above clearly satisfies $\kappa + \delta + \eta + \eta_1 < 1$, and thus \Cref{thm:SDPFeasibility} implies the feasibility of our SDP \eqref{eq:FrameworkSDP} and hence \Cref{fact:sdpFeasibility}.

\section*{Acknowledgements}

We thank Nikhil Bansal and Yuhan Ye for many helpful discussions. 

\bibliographystyle{alphaurl}
\bibliography{bib.bib}

@book{Ver18book,
  title={High-dimensional probability: An introduction with applications in data science},
  author={Vershynin, Roman},
  volume={47},
  year={2018},
  publisher={Cambridge university press}
}

@inproceedings{BJM+21,
  title={Prefix Discrepancy, Smoothed Analysis, and Combinatorial Vector Balancing},
  author={Bansal, Nikhil and Jiang, Haotian and Meka, Raghu and Singla, Sahil and Sinha, Makrand},
  booktitle={13th Innovations in Theoretical Computer Science Conference (ITCS 2022)},
  pages={13--1},
  year={2022},
  organization={Schloss Dagstuhl--Leibniz-Zentrum f{\"u}r Informatik}
}

@inproceedings{BJSS20,
  title={Online vector balancing and geometric discrepancy},
  author={Bansal, Nikhil and Jiang, Haotian and Singla, Sahil and Sinha, Makrand},
  booktitle={Proceedings of the 52nd Annual ACM SIGACT Symposium on Theory of Computing},
  pages={1139--1152},
  year={2020}
}

@article{JKS19,
  title={Online geometric discrepancy for stochastic arrivals with applications to envy minimization},
  author={Jiang, Haotian and Kulkarni, Janardhan and Singla, Sahil},
  journal={arXiv preprint arXiv:1910.01073},
  year={2019}
}

@incollection{Cho94,
  title={Convergence as of rearranged random series in Banach space and associated inequalities},
  author={Chobanyan, Sergej},
  booktitle={Probability in Banach Spaces, 9},
  pages={3--29},
  year={1994},
  publisher={Springer}
}

@inproceedings{HS14,
  title={Near-optimal herding},
  author={Harvey, Nick and Samadi, Samira},
  booktitle={Conference on Learning Theory},
  pages={1165--1182},
  year={2014},
  organization={PMLR}
}

@inproceedings{Ber31,
  title={Zwei s{\"a}tze {\"u}ber ebene vektorpolygone},
  author={Bergstr{\"o}m, Viktor},
  booktitle={Abhandlungen aus dem Mathematischen Seminar der Universit{\"a}t Hamburg},
  volume={8},
  number={1},
  pages={206--214},
  year={1931},
  organization={Springer}
}

@article{HA89,
  title={Bibliography: Series of vectors and Riemann sums},
  author={Halperin, Israel},
  journal={(No Title)},
  year={1989}
}

@inproceedings{JR19,
  title={On integer programming and convolution},
  author={Jansen, Klaus and Rohwedder, Lars},
  booktitle={10th Innovations in theoretical computer science conference (ITCS 2019)},
  year={2018}
}

@article{EW19,
  title={Proximity results and faster algorithms for integer programming using the Steinitz lemma},
  author={Eisenbrand, Friedrich and Weismantel, Robert},
  journal={ACM Transactions on Algorithms (TALG)},
  volume={16},
  number={1},
  pages={1--14},
  year={2019},
  publisher={ACM New York, NY, USA}
}

@article{Sev94,
  title={On some geometric methods in scheduling theory: a survey},
  author={Sevast'janov, Sergey Vasil'evich},
  journal={Discrete Applied Mathematics},
  volume={55},
  number={1},
  pages={59--82},
  year={1994},
  publisher={Elsevier}
}

@article{DFG12,
  title={The master equality polyhedron with multiple rows},
  author={Dash, Sanjeeb and Fukasawa, Ricardo and G{\"u}nl{\"u}k, Oktay},
  journal={Mathematical programming},
  volume={132},
  number={1},
  pages={125--151},
  year={2012},
  publisher={Springer}
}

@article{BMMP12,
  title={Vectors in a box},
  author={Buchin, Kevin and Matou{\v{s}}ek, Ji{\v{r}}{\'\i} and Moser, Robin A and P{\'a}lv{\"o}lgyi, D{\"o}m{\"o}t{\"o}r},
  journal={Mathematical programming},
  volume={135},
  number={1},
  pages={323--335},
  year={2012},
  publisher={Springer}
}

@article{AB86,
  title={Regular hypergraphs, Gordon's lemma, Steinitz'lemma and invariant theory},
  author={Alon, Noga and Berman, Kenneth A},
  journal={Journal of Combinatorial Theory, Series A},
  volume={43},
  number={1},
  pages={91--97},
  year={1986},
  publisher={Elsevier}
}

@article{Ban87,
  title={The Steinitz constant of the plane.},
  author={Banaszczyk, Wojciech},
  year={1987},
  publisher={Walter de Gruyter, Berlin/New York Berlin, New York}
}

@article{GS80,
  title={Value of the Steinitz constant},
  author={Grinberg, Victor S and Sevast'yanov, Sergey V},
  journal={Functional Analysis and Its Applications},
  volume={14},
  number={2},
  pages={125--126},
  year={1980},
  publisher={Springer}
}

@article{Bar81,
  title={A vector-sum theorem and its application to improving flow shop guarantees},
  author={B{\'a}r{\'a}ny, Imre},
  journal={Mathematics of Operations Research},
  volume={6},
  number={3},
  pages={445--452},
  year={1981},
  publisher={INFORMS}
}

@article{Spe77,
  title={Balancing games},
  author={Spencer, Joel},
  journal={Journal of Combinatorial Theory, Series B},
  volume={23},
  number={1},
  pages={68--74},
  year={1977},
  publisher={Elsevier}
}

@inproceedings{BJ26,
  title={Decoupling via Affine Spectral-Independence: Beck-Fiala and Koml$\backslash$'os Bounds Beyond Banaszczyk},
  author={Bansal, Nikhil and Jiang, Haotian}, 
  booktitle={Symposium on Theory of Computing, STOC},
  year={2026},
  organization={ACM}
}

@article{Lev05,
  title={Sur les s{\'e}ries semi-convergentes},
  author={L{\'e}vy, Paul},
  journal={Nouvelles annales de math{\'e}matiques: journal des candidats aux {\'e}coles polytechnique et normale},
  volume={5},
  pages={506--511},
  year={1905}
}

@article{Ste13,
  title={Bedingt konvergente Reihen und konvexe Systeme.},
  author={Steinitz, Ernst},
  year={1913},
  publisher={Walter de Gruyter, Berlin/New York Berlin, New York}
}

@article{BG81,
  title={On some combinatorial questions in finite-dimensional spaces},
  author={B{\'a}r{\'a}ny, Imre and Grinberg, Victor S},
  journal={Linear Algebra and its Applications},
  volume={41},
  pages={1--9},
  year={1981},
  publisher={Elsevier}
}

@article{Ban12,
  title={On series of signed vectors and their rearrangements},
  author={Banaszczyk, Wojciech},
  journal={Random Structures \& Algorithms},
  volume={40},
  number={3},
  pages={301--316},
  year={2012},
  publisher={Wiley Online Library}
}

@inproceedings{BJ25b,
  title={An improved bound for the {Beck-Fiala} conjecture},
  author={Bansal, Nikhil and Jiang, Haotian},
  booktitle={66th IEEE Symposium on Foundations of Computer Science (FOCS)},
  year={2025},
url = {https://arxiv.org/abs/2508.01937},
organization={IEEE}
}

@article{Bar08,
  title={On the power of linear dependencies},
  author={B{\'a}r{\'a}ny, Imre},
  journal={Building Bridges: Between Mathematics and Computer Science},
  pages={31--45},
  year={2008},
  publisher={Springer}
}

@article{Nik17,
  title={Tighter bounds for the discrepancy of boxes and polytopes},
  author={Nikolov, Aleksandar},
  journal={Mathematika},
  volume={63},
  number={3},
  pages={1091--1113},
  year={2017},
  publisher={Wiley Online Library}
}

@inproceedings{BRS22,
  title={Flow time scheduling and prefix {Beck-Fiala}},
  author={Bansal, Nikhil and Rohwedder, Lars and Svensson, Ola},
  booktitle={Symposium on Theory of Computing, STOC},
  pages={331--342},
  year={2022}
}

@article{Ban24,
  title={On a Generalization of Iterated and Randomized Rounding},
  author={Bansal, Nikhil},
  journal={Theory of Computing},
  volume={20},
  number={1},
  pages={1--23},
  year={2024},
  publisher={Theory of Computing Exchange}
}

@inproceedings{BJ25a,
  title={{Quasi-Monte Carlo Beyond Hardy-Krause}},
  author={Bansal, Nikhil and Jiang, Haotian},
  booktitle={ Symposium on Discrete Algorithms (SODA)},
  pages={2051--2075},
  year={2025},
  organization={SIAM}
}

@book{CST14,
  title={A panorama of discrepancy theory},
  author={Chen, William and Srivastav, Anand and Travaglini, Giancarlo},
  volume={2107},
  year={2014},
  publisher={Springer}
}

@inproceedings{MN15,
  author       = {Jir{\'{\i}} Matousek and
                  Aleksandar Nikolov},
  title        = {Combinatorial Discrepancy for Boxes via the $\gamma_2$ Norm},
  booktitle    = {Symposium on Computational Geometry, {SoCG}},
    pages        = {1--15},
  year         = {2015}
}

@inproceedings{Ban22,
  title={Discrepancy theory and related algorithms},
  author={Bansal, Nikhil},
  booktitle={Proc. Int. Cong. Math},
  volume={7},
  pages={5178--5210},
  year={2022}
}

@inproceedings{DGLN16,
  author       = {Daniel Dadush and
                  Shashwat Garg and
                  Shachar Lovett and
                  Aleksandar Nikolov},
  title        = {Towards a Constructive Version of {B}anaszczyk's Vector Balancing Theorem},
  booktitle    = { {APPROX/RANDOM} 2016},
  pages        = {28:1--28:12},
  year         = {2016},
  url          = {https://doi.org/10.4230/LIPIcs.APPROX-RANDOM.2016.28},
  doi          = {10.4230/LIPICS.APPROX-RANDOM.2016.28}
}

@inproceedings{JSS23,
  author       = {Vishesh Jain and
                  Ashwin Sah and
                  Mehtaab Sawhney},
  title        = {Spencer's theorem in nearly input-sparsity time},
  booktitle    = { Symposium on Discrete Algorithms,
                  {SODA}},
  pages        = {3946--3958},
year={2023},
  url          = {https://doi.org/10.1137/1.9781611977554.ch152},
  doi          = {10.1137/1.9781611977554.CH152}}

@inproceedings{LRR17,
  title={Deterministic discrepancy minimization via the multiplicative weight update method},
  author={Levy, Avi and Ramadas, Harishchandra and Rothvoss, Thomas},
  booktitle={Integer Programming and Combinatorial Optimization (IPCO)},
  pages={380--391},
  year={2017}
}

@article{BS13,
  title={Deterministic discrepancy minimization},
  author={Bansal, Nikhil and Spencer, Joel},
  journal={Algorithmica},
  volume={67},
  pages={451--471},
  year={2013},
  publisher={Springer}
}

@book{Cha00,
  title={The Discrepancy Method: Randomness and Complexity},
  author={Chazelle, Bernard},
  year={2000},
  publisher={Cambridge University Press}
}

@book{Mat09,
  title={Geometric discrepancy: An illustrated guide},
  author={Matousek, Jiri},
  volume={18},
  year={2009},
  publisher={Springer Science \& Business Media}
}

@inproceedings{Ban10,
  title={Constructive algorithms for discrepancy minimization},
  author={Bansal, Nikhil},
  booktitle={Symposium on Foundations of Computer Science},
  pages={3--10},
  year={2010}
}

@article{Spe85,
  title={Six standard deviations suffice},
  author={Spencer, Joel},
  journal={Transactions of the American mathematical society},
  volume={289},
  number={2},
  pages={679--706},
  year={1985}
}

@article{Glu89,
  title={Extremal properties of orthogonal parallelepipeds and their applications to the geometry of Banach spaces},
  author={Gluskin, Efim Davydovich},
  journal={Mathematics of the USSR-Sbornik},
  volume={64},
  number={1},
  pages={85},
  year={1989}
}

@article{Rot17,
  title={Constructive discrepancy minimization for convex sets},
  author={Rothvoss, Thomas},
  journal={SIAM Journal on Computing},
  volume={46},
  number={1},
  pages={224--234},
  year={2017},
  publisher={SIAM}
}

@article{LM15,
  title={Constructive discrepancy minimization by walking on the edges},
  author={Lovett, Shachar and Meka, Raghu},
  journal={SIAM Journal on Computing},
  volume={44},
  number={5},
  pages={1573--1582},
  year={2015},
  publisher={SIAM}
}

@inproceedings{HSS14,
  title={Discrepancy without partial colorings},
  author={Harvey, Nicholas and Schwartz, Roy and Singh, Mohit},
  booktitle={ {APPROX/RANDOM}},
  year={2014}
}

@article{BF81,
  title={“{I}nteger-making” theorems},
  author={Beck, J{\'o}zsef and Fiala, Tibor},
  journal={Discrete Applied Mathematics},
  volume={3},
  number={1},
  pages={1--8},
  year={1981},
  publisher={Elsevier}
}

@article{Ban98,
  title={Balancing vectors and {G}aussian measures of n-dimensional convex bodies},
  author={Banaszczyk, Wojciech},
  journal={Random Structures \& Algorithms},
  volume={12},
  number={4},
  pages={351--360},
  year={1998},
  publisher={Wiley Online Library}
}

@inproceedings{BLV22,
  title={A Unified Approach to Discrepancy Minimization},
  author={Bansal, Nikhil and Laddha, Aditi and Vempala, Santosh},
      booktitle={ {APPROX/RANDOM}},
  pages={1--1},
  year={2022}
}

@inproceedings{BG17,
  title={Algorithmic discrepancy beyond partial coloring},
  author={Bansal, Nikhil and Garg, Shashwat},
  booktitle={ Symposium on Theory of Computing, {STOC}},
  pages={914--926},
  year={2017}
}

@article{BDG19,
  title={An algorithm for {K}oml{\'o}s conjecture matching Banaszczyk's bound},
  author={Bansal, Nikhil and Dadush, Daniel and Garg, Shashwat},
  journal={SIAM Journal on Computing},
  volume={48},
  number={2},
  pages={534--553},
  year={2019},
  publisher={SIAM}
}

@article{ES18,
  author    = {Ronen Eldan and
               Mohit Singh},
  title     = {Efficient algorithms for discrepancy minimization in convex sets},
  journal   = {Random Struct. Algorithms},
  volume    = {53},
  number    = {2},
  pages     = {289--307},
  year      = {2018}
}

@article{HSSZ24,
  title={Balancing covariates in randomized experiments with the {G}ram--{S}chmidt walk design},
  author={Harshaw, Christopher and S{\"a}vje, Fredrik and Spielman, Daniel A and Zhang, Peng},
  journal={Journal of the American Statistical Association},
  volume={119},
  number={548},
  pages={2934--2946},
  year={2024},
  publisher={Taylor \& Francis}
}

@inproceedings{PV23,
  author       = {Lucas Pesenti and
                  Adrian Vladu},
  title        = {Discrepancy Minimization via Regularization},
  booktitle    = {Symposium on Discrete Algorithms,
                  {SODA}},
  pages        = {1734--1758},
  year         = {2023},
  doi          = {10.1137/1.9781611977554.CH66}
}

@inproceedings{ALS21,
    AUTHOR = {Alweiss, Ryan and Liu, Yang P. and Sawhney, Mehtaab},
     TITLE = {Discrepancy minimization via a self-balancing walk},
 BOOKTITLE = {
              Symposium on Theory of Computing, {STOC}},
     PAGES = {14--20},
      YEAR = {2021},
      ISBN = {978-1-4503-8053-9},
       DOI = {10.1145/3406325.3450994},
       URL = {https://doi.org/10.1145/3406325.3450994}
}

@inproceedings{BDGL18,
  title={The {Gram-Schmidt} walk: a cure for the { Banaszczyk} blues},
  author={Bansal, Nikhil and Dadush, Daniel and Garg, Shashwat and Lovett, Shachar},
  booktitle={Symposium on Theory of Computing, {STOC}},
  pages={587--597},
  year={2018}
}

\end{document}